\documentclass[twocolumn,english,groupedaddress, superscriptaddress,aps,pre,10pt,floatfix]{revtex4-2}
\usepackage[T1]{fontenc}
\usepackage[utf8]{inputenc}
\usepackage{amsmath,amsthm,mathtools,amssymb}
\usepackage{graphicx}
\usepackage{units}
\usepackage{dsfont}

\usepackage{url}

\usepackage[breaklinks]{hyperref}

\newtheorem{cor}{Corollary}
\newtheorem{lem}{Lemma}
\newtheorem{thm}{Theorem}

\newcommand{\matr}[1]{{\boldsymbol{#1}}}
\renewcommand{\vec}[1]{{\boldsymbol{#1}}}
\newcommand{\EE}{\mathfrak{E}}
\newcommand{\VV}{\mathfrak{V}}

\newcommand{\GG}{\mathfrak{G}}
\newcommand{\LL}{\mathcal{L}}
\newcommand{\OO}{\mathcal{F}}
\newcommand{\GC}{\mathcal{G}}
\newcommand{\nn}{N}
\newcommand{\nl}{M}
\newcommand{\eye}{\mathds{1}}

\begin{document}

\title{Synchronized states of power grids and oscillator networks by convex optimization}

\author{Carsten Hartmann}%
 \email{c.hartmann@fz-juelich.de}
  \affiliation{Institute of Energy and Climate Research -- Energy Systems Engineering (IEK-10), Forschungszentrum J\"ulich,  52428 J\"ulich, Germany}
  
\author{Philipp C. Böttcher}%
  \affiliation{Institute of Energy and Climate Research -- Energy Systems Engineering (IEK-10), Forschungszentrum J\"ulich,  52428 J\"ulich, Germany}

\author{David Gross}%
\affiliation{Institute for Theoretical Physics, University of Cologne, 50937 K\"oln, Germany}

\author{Dirk Witthaut}%
 \email{d.witthaut@fz-juelich.de}
  \affiliation{Institute of Energy and Climate Research -- Energy Systems Engineering (IEK-10), Forschungszentrum J\"ulich,  52428 J\"ulich, Germany}
  \affiliation{Institute for Theoretical Physics, University of Cologne, 50937 K\"oln, Germany}

\begin{abstract}
Synchronization is essential for the operation of AC power systems: All generators in the power grid must rotate with fixed relative phases to enable a steady flow of electric power. Understanding the conditions for and the limitations of synchronization is of utmost practical importance.  
In this article, we propose a novel approach to compute and analyze the stable stationary states of a power grid or an oscillator network in terms of a convex optimization problem. This approach allows to systematically compute \emph{all} stable states where the phase difference across an edge does not exceed $\pi/2$.
Furthermore, the optimization formulation allows to rigorously establish certain properties of synchronized states and to bound the error in the widely used linear power flow approximation.
\end{abstract}
  
\maketitle

\section{Introduction}

A reliable supply of electricity is vital for our society. The operation of the electric power system relies on a stable, synchronized state of the grid: All generators operate at the same frequency close to the reference frquency of of 50 or 60 Hz, respectively~\cite{dorfler2013synchronization}. The strict phase locking enables a steady flow of electric power over large distances~\cite{witthaut2022collective}. Violations of synchronization can occur after failures, disrupting power flows and eventually resulting in widespread power outages (see, e.g.,~\cite{UCTE07}). Remarkably, similar problems occur in various systems studied in statistical physics. The celebrated Kuramoto model describes the dynamics of coupled limit-cycle oscillators and bears strong similarities to models of coupled generators~\cite{kuramoto1975self,acebron2005kuramoto,rodrigues2016kuramoto}. 

Due to their utmost importance, synchronization and stability are essential topics in various disciplines~\cite{dorfler2014synchronization}. 
A central question is: Does a given network support a stable, fully synchronous state or not? Clearly, a strong coupling fosters synchronization, while a high power load impedes synchronization -- but a comprehensive rigorous answer to this question is still lacking. For general networks, accurate sufficient conditions have been developed in~\cite{dorfler2013synchronization} and sharpened in Ref.~\cite{jafarpour2018synchronization}, while the role of the network topology was explored in Refs.~\cite{rohden2012self,motter2013spontaneous}. Moreover, it has been realized that sparse networks can be multistable, i.e., they support more than one stable synchronized state~\cite{wiley2006size}. Again we are left to the question of when these states exist~\cite{delabays2016multistability,manik2017cycle}.
Another important question is how vulnerable a given synchronous state is to perturbations\cite{menck2014dead,delabays2017size, tyloo2018robustness}.

In this article, we propose an alternative approach to the synchronization problem by mapping the fixed point equations to a convex optimization problem. In the case of potential multistability, each synchronous state is associated with an individual convex problem. This reformulation allows to systematically compute all synchronized states or to refute their existence. The optimization formulation provides further insights into the analytic and geometric properties of synchronized states. For instance, we provide an alternative view of the linear power flow or DC approximation widely used in power engineering~\cite{machowski2020power}. The formulation as an optimization problem allows us to derive rigorous bounds on the error emerging from the linearization. 

The article is organized as follows.
We first review essential models for power grid synchronization as well as the Kuramoto model in Sec.~\ref{sec:background}.
The mathematical background and important prior results on the subject are summarized in Sec.~\ref{sec:notation-review}.
The main result -- the optimization approach -- is formulated in Sec.~\ref{sec:powerflow-opt}.
The relation to the linear or DC power flow equations is analyzed in Sec.~\ref{sec:DCapprox}.
Geometric aspects of the power flow optimization problem are discussed in Sec.~\ref{sec:geometry}.
We close with a summary and an outlook to future research in Sec.~\ref{sec:conclusion}.

\section{Stationary states in power grids and oscillator networks}
\label{sec:background}

In this section, we review models for the dynamics of electric power grids and coupled limit cycle oscillators. In all cases, the stationary state is described by a set of nonlinear algebraic equations which is the central objective of this article.

\subsection{The swing equation and the classical model}
\label{sec:swing}

Synchronous generators traditionally form the backbone of electric power generation. The state of the $n$th machine is characterized by the mechanic and electric phase angle $\theta_n(t)$, using a frame of reference rotating at the nominal frequency of the grid. 
The phase angle evolves according to the \emph{swing equation}~\cite{machowski2020power}
\begin{equation}
   M_n \ddot \theta_n + D_n \dot \theta_n = p_n^{\rm mech} - p_n^{\rm el} \, .
\end{equation}
Here, $p_n^{\rm mech}(t)$ is the mechanical input power to the generator, $p_n^{\rm el}(t)$ the power delivered to the grid, $M_n$ is an inertia constant, and $D_n$ a damping constant. 

Furthermore, a model for the load nodes is needed to close the equations of motion~\cite{witthaut2022collective}. One classical approach assumes load nodes to be fully passive and describes them as constant impedances to the ground. These nodes can be eliminated from the network equations, a process commonly referred to as Kron reduction~\cite{dorfler2012kron,nishikawa2015comparative}. The resulting model includes only generator nodes and the electric power exchanges with the grid can be written as
\begin{align*}
    p_n^{\rm el}(t) = \Delta p^{\rm eff}_n + \sum_{m} K^{\rm eff}_{nm} \sin(\theta_n(t) - \theta_m(t)) \, ,
\end{align*}
with an effective connectivity $K^{\rm eff}_{nm}$ and an effective power injection $\Delta p^{\rm eff}_n$ accounting for the eliminated nodes.
The stationary states of the grid are thus given by the nonlinear algebraic equations
\begin{align}
    p^{\rm eff}_n = 
    \sum_{m} K^{\rm eff}_{nm} \sin(\theta_n - \theta_m)
    \quad \forall \;
    \mbox{generator nodes} \; n
    \label{eq:stationary-classical}
\end{align}
with $p^{\rm eff}_n=p_n^{\rm mech}- \Delta p^{\rm eff}_n$. Notably, the coupling matrix $K^{\rm eff}_{nm}$ now describes an effective network, not the physical power transmission grid. This matrix is typically full, i.e.,~all entries are non-zero~\cite{dorfler2012kron}.

\subsection{The structure-preserving model}

A different approach that conserves the network structure has been suggested by Bergen and Hill \cite{bergen1981structure}. The starting point is the observation that the real power drawn from the grid by a load node $n$ typically increases with the frequency. Linearizing around the nominal frequency, the load is thus written as
\begin{equation}
    p_n^{\rm demand} =  p_n^{\rm demand, 0} + D_n \dot \theta_n . 
\end{equation}
For a consistent description, we need the power injected into the grid, which is simply given by $p_n^{\rm el} =  -p_n^{\rm demand}$. Then, we can write the evolution equation for \emph{all} nodes as
\begin{equation}
   M_n \ddot \theta_n + D_n \dot \theta_n = p_n^{\rm in} - p_n^{\rm el},
\end{equation}
with $M_n = 0$ and $p_n^{\rm in} = -p_n^{\rm demand, 0}$ for load nodes and $p_n^{\rm in} = +p_n^{\rm mech}$ for generator nodes. The real power flow between two arbitrary nodes $n$ and $m$ is given 
\begin{equation}
   p_{n \rightarrow m} = K_{nm} \sin(\theta_n - \theta_m).
\end{equation}
Hence, the equations of motion read
\begin{equation}
   M_n \ddot \theta_n + D_n \dot \theta_n = p_n^{\rm in} - \sum_m K_{nm} \sin(\theta_n - \theta_m).
\end{equation}
and the stationary states are given by the nonlinear algebraic equations
\begin{align}
    p^{\rm in}_n = 
    \sum_{m} K_{nm} \sin(\theta_n - \theta_m)
    \quad \forall \;
    \mbox{nodes} \; n.
    \label{eq:stationary-preserving}
\end{align}
We emphasize that this model includes all nodes (buses) and all edges (branches) of the network. Neglecting Ohmic losses, the coupling constants $K_{nm}$ are proportional to the inverse of the reactance $X_{nm}^{-1}$ for a transmission line $(n,m)$ and $K_{nm} = 0$ if no such line exists. Hence, the model conserves the structure of the actual power grid, which is sparse, and we may interpret $K_{nm}$ as a weighted adjacency matrix.

\subsection{The Kuramoto model}

The Kuramoto model describes the dynamics of coupled limit-cycle oscillators with applications in various scientific disciplines \cite{kuramoto1975self,acebron2005kuramoto,rodrigues2016kuramoto}. 
The position of an oscillator $n$ along the limit cycle is encoded in a phase variable $\theta_n(t)$ which evolves according to the equations of motion
\begin{align*}
    \frac{d \theta_n}{dt} = 
    \omega^{(0)}_n + \sum_{m=1}^{N} K_{nm} \sin(\theta_m - \theta_n),
\end{align*}
where $\omega^{(0)}_n$ denotes the natural frequency of the oscillator $n$, $K_{nm} = K_{mn} \ge 0$ is a coupling constant and $N$ denotes the number of oscillators in the network. 

Traditionally, research on the Kuramoto model has focused on large systems ($N \rightarrow \infty$) and the transition from incoherence (all oscillators rotate incoherently with $\dot \theta_n \approx \omega^{(0)}_n$) to partial synchronization (a subset of oscillators lock their frequencies $\dot \theta_n$). An excellent introduction to this direction of research is given in Ref.~\cite{strogatz2000kuramoto}.

More recently, finite systems and their fixed points have received increased attention. The fixed points are determined by the set of equations
\begin{align*}
    \omega^{(0)}_n = \sum_{m=1}^{N} K_{nm} \sin(\theta_n - \theta_m), \qquad \forall \, n \in \{1,\ldots,N\},
\end{align*}
which is equivalent to equations \eqref{eq:stationary-classical} and \eqref{eq:stationary-preserving} if we identify $\omega^{(0)}_n$ with $p_n$.
Two lines of mathematical research deserve to be mentioned: First, several authors exploit the equivalence to flow networks \cite{dorfler2013synchronization,delabays2016multistability,delabays2017multistability,manik2017cycle,balestra2019multistability}.
Second, the fixed point equation can be complexified to remove the sine \cite{muller2021algebraic,chen2022toric}.

\subsection{Stability of fixed points}

For all of the above models, one can show that a stationary state is stable if and only if the Laplacian matrix $\matr L$ with elements
\begin{align*}
    L_{m,n} = \left\{ \begin{array}{lll}
    - K_{nm} \cos(\theta_n - \theta_m) & 
         \; \mbox{for} \; &
         n\neq m \\
    \sum_l K_{nl} \cos(\theta_n - \theta_l) & &
    n = m
    \end{array} \right.
\end{align*}
has only positive eigenvalues except for a trivial zero eigenvalue corresponding to the eigenvector $(1,1,\ldots,1)^\top$~\cite{manik2014supply,witthaut2022collective}. A sufficient condition for stability is that
\begin{align}
    \cos(\theta_n - \theta_m) > 0
    \label{eq:cond-normal}
\end{align}
for all edges $(n,m)$ in the network. Stationary states that violate this condition are typically, but not always, unstable~\cite{manik2014supply}. In the following, we will focus on the stationary states that satisfy condition~\eqref{eq:cond-normal} for all edges and will refer to them as ``normal'' states or solutions.

\subsection{The linear power flow or DC approximation}
\label{intro-DC}

In power transmission grids, the phase difference $\theta_n-\theta_m$ along a transmission line $(n,m)$ is typically small. Hence, it is often appropriate to linearize the sine in the expression for the real power flow. The steady state of a lossless power grid is then determined by the linear set of equations
\begin{align}
    p_n = \sum_m K_{nm} \left( \theta_n - \theta_m \right).
\end{align}
This approximation is referred to as the \emph{linear power flow} or the \emph{DC approximation} because the resulting equations are mathematically equivalent to Kirchhoff’s equations for DC electric circuits. The validity of the linear power flow approximation has been studied numerically in Refs.~\cite{purchala2005usefulness,stott2009dc}.

\section{Mathematical Background}
\label{sec:notation-review}

In this section, we turn to the mathematical analysis of the equations
\begin{align}
    p_n = \sum_m K_{nm} \sin(\theta_n - \theta_m) \qquad
    \forall \; \mbox{nodes} \; n
    \label{eq:realpower-intro}
\end{align}
that describe the stationary state of a power grid or a set of Kuramoto oscillators. These will be referred to as the \emph{real power flow equations} in the following to distinguish from linear power flow or the load-flow equations that also include the reactive power~\cite{witthaut2022collective}. We start with general aspects of (linear) flow networks before we move to the nonlinear real power flow.

\subsection{Flow networks and algebraic graph theory}

We introduce some useful notation and review some important results from the literature. So consider a network with $\nn$ vertices or nodes and $\nl$ transmission lines or edges. Nodes are labeled as $n = 1,\ldots,\nn$, and the node set is denoted as $\VV$. Throughout this article, we assume that the network is connected. Edges are either labeled by their endpoints $(n,m)$ or consecutively as $e = 1,\ldots,\nl$, and the set of edges is denoted as $\EE$. To keep track of the direction of power flows we fix an orientation for each edge.  That is, if the line $e \equiv (n,m)$ is oriented from $n$ to $m$, then a positive flow points from $n$ to $m$ and a negative flow points from $m$ to $n$. To keep track of the structure of the topology of the grid and the orientation, we define  the node-edge incidence matrix $\matr E \in \mathbb{R}^{\nn \times \nl}$ with components \cite{Newm10}
\begin{equation}
   E_{n,e} = \left\{
   \begin{array}{r l}
      1 & \; \mbox{if line $e$ starts at node $n$},  \\
      - 1 & \; \mbox{if line $e$ ends at node $n$},  \\
      0     & \; \mbox{otherwise}.
  \end{array} \right.
  \label{eq:def-nodeedge}
\end{equation}
The edges' transmission capacities $K_e = K_{nm}$ are summarized in a diagonal matrix $\matr K = \mbox{diag}(K_1,\ldots,K_\nl)$.

A node $n$ is characterized by the power injection $p_n$  or its natural frequency, respectively. All power injections are summarized in the vector $\vec p = (p_1,\ldots,p_\nn)^\top$, where the superscript ${}^\top$ denotes the transpose of a vector or matrix. As we are dealing with lossless power grids, we will assume $\sum_n p_n = 0$ from now on. Furthermore, we define the vector of nodal phase angles $\vec \theta = (\theta_1,\ldots,\theta_\nn)^\top$. The real power flow of an edge $e$ is denoted as $f_e$, and all flows are summarized in the vector $\vec f = (f_1,\ldots,f_\nl)^\top$.

For lossless power grids, the real power flows have to satisfy the continuity equation or Kirchhoff's current law (KCL)
\begin{align}
    p_n = \sum_{m=1}^{\nn} f_{nm} , \qquad \forall n \in \{1,\ldots,\nn\}.
\end{align}
This set of equations can also be written in a vectorial form as
\begin{equation}
    \vec p = \matr E \vec f . 
    \label{eq:KCL}
\end{equation}

The problem admits a straightforward solution if we impose the linear power flow or DC approximation. The flows now read 
\begin{align*}
    f_{nm} &= K_{nm}  (\theta_n - \theta_m), \\
    \vec f &=  \matr K  \matr E^\top \vec \theta .
\end{align*}
Substituting the last equation into the KCL yields
\begin{align}
    \vec p = \matr E \matr K  \matr E^\top \vec \theta  = \matr L \vec \theta
    \label{eq:DC-pLt}
\end{align}
where we have introduced the Laplacian matrix $\matr L \in \mathbb{R}^{\nn \times \nn}$ with components \cite{Newm10}
\begin{equation}
   L_{n,m} = \left\{
   \begin{array}{r l}
         \sum_{l=1}^\nn K_{nl}  & \; \mbox{for} \, n = m, \\
         - K_{nm} & \; \mbox{if $n\neq m$ and $n$, $m$ are adjacent},  \\
      0     & \; \mbox{otherwise}.
  \end{array} \right.
  \label{eq:def-Laplacian}
\end{equation}
The equation \eqref{eq:DC-pLt} is linear and can thus be solved for $\theta$ in a straightforward way,
\begin{align*}
    \vec \theta &= \matr L^+ \vec p,  \\
    \vec f &= \matr K  \matr E^\top  \matr L^+ \vec p.
\end{align*}
where the superscript $+$ denotes the Moore-Penrose pseudoinverse.

We note that the equations above still include a gauge freedom: If we shift all phases $\theta_n$ by a constant, then the power flows will not change. This is reflected by the fact that the Laplacian matrix of a connected network has rank $N-1$. A typical choice in applications is to fix the phase of one distinguished slack node as $\theta_n \equiv 0$ and remove this node from the set of equations \eqref{eq:DC-pLt}. Removing the respective row and column from $\matr L$, we obtain a \emph{grounded} Laplacian matrix.

\subsection{Cycle flows and Helmholtz decomposition}

Before we proceed with the real power flow, we provide some further results on network flows. We recall that we assume that the network is connected throughout this article.

We first note that the KCL \eqref{eq:KCL} is underdetermined such that the general solution can be written as 
\begin{align*}
    \vec f = \vec f^{(0)} + \vec f^{(c)} ,
\end{align*} 
where $\vec f^{(0)}$ is a special solution and $\vec f^{(c)}$ denotes an arbitrary solution of the associated homogeneous equation, i.e., a vector from the kernel of the node-edge incidence matrix $\matr E$. The vectors $\vec f^{(c)}$ correspond to cycle flows: The power balance $\matr E \vec f^{(c)}$ vanishes everywhere, implying that no real power flows in or out of the network. We can fix a basis for the kernel by choosing  $\nl-\nn+1$ independent fundamental cycles and encode this basis in the cycle-edge incidence matrix $\matr C \in \mathbb{R}^{\nl \times (\nl-\nn+1)}$ with components 
\begin{equation}
   C_{e,\beta} = \left\{
   \begin{array}{r l}
      1 & \; \mbox{if edge $e$ belongs to cycle $\beta$},  \\
      - 1 & \; \mbox{if the reverse edge $e$ belongs to cycle $\beta$},  \\
      0     & \; \mbox{otherwise},
  \end{array} \right.
  \label{eq:def-cyclein}
\end{equation}
such that $\matr{E} \matr{C} = \matr{0}$. An important example is given by plane networks, i.e.,~networks that are drawn in the plane without any edge crossing. Here one can choose the faces or plaquettes of the graph as fundamental cycles. 
After fixing a basis, any cycle flow vector can be written as
$$
   \vec f^{(c)} = \matr C \vec \ell,
$$ 
where $\vec \ell = (\ell_1,\ldots,\ell_{\nl-\nn+1})$ is a vector of cycle or loop flow amplitudes. We note that this decomposition proves to be useful in various linear power flow problems~\cite{ronellenfitsch2017dual,horsch2018linear}.

Similar to the Helmholtz decomposition in vector analysis, a flow $\vec f$ can be decomposed into a cycle or source-free component $\vec f^{(c)} \in \ker{( \matr E)}$ and a directed or irrotational component $\vec f^{(d)}$ as we make precise in the following lemma. We assume that the flow $\vec f$ is energy conserving, i.e.,~the sum of nodal out- and inflows vanishes, $\sum_n (\matr E \vec f)_n = 0$. 

\begin{lem}
\label{lem:helmholtz}
Any energy-conserving flow vector can be decomposed as 
\begin{equation}
    \vec f = \vec f^{(d)} + \vec f^{(c)} \, ,
\end{equation}
where $\vec f^{(c)} \in \ker{( \matr E)}$ is a pure cycle flow and
$\vec f^{(d)} \in \ker{( \matr C^\top \matr K^{-1})}$ by applying the projections
\begin{align}
    \begin{split}
        \matr \Pi _{\rm dir} &= \matr K\matr E^\top \matr L^{+} \matr E\\
        \matr \Pi _{\rm cycle} &= \eye - \matr\Pi _{\rm dir} \, .
    \end{split}
    \label{eq:def-projector}
\end{align}
The projectors are orthogonal with respect to the inner product
\begin{equation}
    \langle \vec \xi, \vec \zeta \rangle_K := 
       \sum_{e=1}^M \xi_e \, K_e^{-1} \, \zeta_e \, .
       \label{eq:def-inner-K}
\end{equation}
\end{lem}
\begin{proof} 
We first show that the two matrices are projections and span the entire space:
\begin{enumerate}
        \item The map $\matr \Pi _{\rm dir}$ is a projection, that is 
        \begin{align*}
            \matr \Pi _{\rm dir}^2 &= \matr K\matr E^\top \matr L^{+} \underbrace{\matr E \matr K \matr E^\top}_{=\matr L} \matr L^{+} \matr E \\
            & = \matr K\matr E^\top \matr L^{+} \matr E 
            = \matr \Pi _{\rm dir}.
        \end{align*}
        We here use that the Moore-Penrose pseudo inverse is a weak inverse in the sense that it satisfies $\matr L^+ \matr L \matr L^+  = \matr L^+$. 
        \item The map $\matr \Pi _{\rm cycle}$ is a projection, that is
        \begin{align*}
            \matr \Pi_{\rm cycle}^2 &= \left( \eye - \matr \Pi_{\rm dir} \right)^2 
            = \eye - 2 \matr \Pi_{\rm dir} + \matr \Pi^2_{\rm dir} \\
            &= \eye - \matr \Pi_{\rm dir} = \matr \Pi_{\rm cycle} \, .
        \end{align*}
        \item 
        Finally, is is easy to see that
        $$
            \matr \Pi_{\rm cycle} + \matr \Pi_{\rm dir} = \eye.
        $$
\end{enumerate}
Second, it is easy to see that the projections are orthogonal with respect to the inner product \eqref{eq:def-inner-K},
\begin{equation}
    \langle \vec \xi, \matr \Pi_{\rm dir} \vec \zeta \rangle_K =
    \vec \xi^\top \matr E^\top \matr L^{+} \matr E \vec \zeta =
    \langle \matr \Pi_{\rm dir} \vec \xi,  \vec \zeta \rangle_K \, 
\end{equation}
where we note that $\matr E^\top \matr L^{+} \matr E$ is real symmetric. \\ 
Finally, we show that $\matr \Pi_{\rm cycle}$ has the desired properties.
\begin{enumerate}
        \item For a pure cycle flow $\vec f^{(c)} \in \ker{( \matr E)}$ we have 
        $$ 
        \matr \Pi _{\rm cycle} \vec f^{(c)} = \vec f^{(c)}. 
        $$
        That is, the projector $\matr \Pi _{\rm cycle}$ leaves any cycle flow invariant as desired.        
        \item
        For an arbitrary flow vector $\vec f$ we have
        \begin{align*}
            \matr E \,  \matr \Pi _{\rm cycle} \vec f 
            &= \matr E \left( \eye -  \matr K\matr E^\top \matr L^{+} \matr E  \right) \vec f  \\
            &= \matr E \vec f - \underbrace{\matr E \matr K \matr E^\top}_{=\matr L} \matr L^{+} \matr E  \vec f.
        \end{align*}
        Multiplying with $\matr L^+$ from the left, we get 
        \begin{align*}
        \matr L^+ \matr E \,  \matr \Pi _{\rm cycle} \vec f &= \matr L^+ \matr E \vec f - \matr L^+ \matr L \matr L^{+} \matr E  \vec f \\
             &= \matr L^+ (\matr E \vec f - \matr E \vec f) = \matr L^+ \vec 0 = \vec 0. 
        \end{align*}        
        Hence, we have found that $\matr E \,  \matr \Pi _{\rm cycle} \vec f \in \ker(\matr L^+)$. As we assume that the network is connected, the kernel of $\matr L^+$ is given by the linear subspace $\{c \cdot (1, ..., 1)^T | c \in \mathbf{R} \}$. We thus have that 
        \begin{equation*}
            \matr E \,  \matr \Pi _{\rm cycle} \vec f = c \cdot (1, ..., 1)^T
        \end{equation*}
        for some $c$. \\
        On the other hand, by demanding that the flow is energy conserving we have that $\sum_n (\matr E \,  \matr \Pi _{\rm cycle} \vec f)_n = 0$. This is only fulfilled if $c=0$, that is 
        \begin{equation*}
            \matr E \,  \matr \Pi _{\rm cycle} \vec f = 0.
        \end{equation*}
        We conclude that the projected flow $\matr \Pi _{\rm cycle} \vec f $ is indeed a cycle flow as it is in the kernel of $\matr E$.    
    \end{enumerate}
\end{proof}

\subsection{Resistance Distance}

A common metric in graph theory is the resistance distance or effective resistance~\cite{randic1993resistance}. It is defined via an analogy to DC electric circuits, replacing every edge $a$ with an Ohmic resistor. In this article, we will assume that the conductance of an edge $a$ is given by the coupling strength $K_a$, i.e.,~the resistance is $K_a^{-1}$. 
The resistance distance $\Omega_{m,n}$ between two nodes $n$ and $m$ is then defined as the voltage drop between $n$ and $m$ divided by the total current between $n$ and $m$. In the context of power grids, resistance distances provide essential information about the robustness of synchronized states~\cite{tyloo2018robustness}.

To compute $\Omega_{m,n}$ we consider a unit current injected at node $n$ and withdrawn at node $m$
\begin{equation*}
    \vec\iota_{n,m} = \vec w_n -\vec w_m ,
\end{equation*}
where $\vec w_n$ is the $n$-th standard basis vector, that is $(\vec w_n)_l = \delta_{n,l}$. \\
The nodal voltages are then given by the vector $\vec u = \matr L^+ \vec \iota_{n,m}$ such that we obtain
\begin{align*}
     \Omega_{n,m} &= \vec \iota_{n,m}^\top \vec u \\
     &= (\vec w_n -\vec w_m)^\top \matr L^+ (\vec w_n -\vec w_m).
\end{align*}
If the nodes $n$ and $m$ are connected by an edge $a = (n,m)$, then we have $\vec w_n -\vec w_m = \matr E \vec w_a$ and thus 
\begin{equation}
    \Omega_{a} = \vec w_a^\top \matr E^\top \matr L^+ \matr E \vec w_a.
    \label{eq:Omega-a}
\end{equation}
Inserting $\eye = \matr K^{-1} \matr K$, and recalling the definition of the projectors (Eq.~\eqref{eq:def-projector}) we get 
\begin{align*}
    \Omega_a &= \vec w_a^\top \matr K^{-1} (\eye - \matr \Pi _{\rm cycle}) \vec w_a \\
    &= K_a^{-1} \left(1 - (\matr \Pi _{\rm cycle})_{a,a}\right). 
\end{align*}
We note that $(\matr \Pi _{\rm cycle})_{a,a} \ge 0$ is an indicator for how many fundamental cycles are incident to the edge $a$. This quantity vanishes if the edge is part of no cycle, i.e.,~if it is a bridge.

\subsection{Algebraic formulation and multistability of the real power flow equations}

We now turn back to the real power flow equation \eqref{eq:realpower-intro}. This problem is harder to tackle than the linear power flow due to the nonlinearity introduced by the sine function. In particular, the flows are now given by
\begin{align}
    f_{nm} &= K_{nm}  \sin(\theta_n - \theta_m), \nonumber \\
    \vec f &= \matr K \sin\left( \matr E^\top \vec \theta \right),
    \label{eq:realpower-ft}
\end{align}
where the sine function is taken element-wise. 
In general, the equations admit multiple solutions such that it is a priori unclear how and to which solution a numerical solver will converge. To systematically compute all solutions, one can proceed in two steps \cite{manik2017cycle}. First, we compute all solution candidates, that is all solutions of the KCL \eqref{eq:KCL}. As described before, the general solution can be written as
\begin{align*}
    \vec f = \vec f^{(0)} + \vec f^{(c)} ,
\end{align*} 
where $\vec f^{(c)}$ is an arbitrary cycle flow. 

Now that we have a set of solution candidates, we have to select the actual solutions of the real power equations \eqref{eq:realpower-intro} from this set. Consider an edge $a \equiv (m,n)$. Inverting equation \eqref{eq:realpower-ft}, we can recover the difference of the nodal phase angles at nodes $m$ and $n$ by
\begin{align*}
    \theta_n - \theta_m &= \arcsin\left(  \frac{f_{nm}}{K_{nm}} \right),  \qquad  \mbox{or} \\
    \theta_n - \theta_m &= \pi - \arcsin\left(  \frac{f_{nm}}{K_{nm}} \right).
\end{align*}
modulo $2 \pi$. The first option leads to normal fixed points of Eq.~\eqref{eq:realpower-intro} that are guaranteed to be stable. The second option typically (but not always) yields unstable solutions, hence they will be discarded from now on.  Inserting the cycle flow decomposition $ \vec f = \vec f^{(0)} + \matr C \vec \ell$ then yields
\begin{align*}
    \theta_n - \theta_m = \arcsin\left(  \frac{f_e^{(0)}  + \sum_{\beta=1}^{\nl-\nn+1} C_{e,\beta} \ell_\beta }{K_e} \right).
\end{align*}
Trying to compute all phase angles typically yields a problem as we cannot satisfy this relation simultaneously for all edges in the grid. Only for discrete values of the loop flow amplitudes $\vec \ell$, we get a consistent solution. It has been shown that this is the case if the condition
\begin{equation}
    \sum_{e=1}^{\nl} C_{e,\alpha} \arcsin\left(  \frac{f_e^{(0)}  + \sum_{\beta=1}^{\nl-\nn+1} C_{e\beta} \ell_\beta }{K_e} \right) = 2 \pi z_\alpha
    \label{eq:cycle-condition}
\end{equation}
with $z_\alpha \in \mathbb{Z}$ is satisfied for all fundamental cycles $\alpha = 1,\ldots,\nl-\nn+1$~\cite{manik2017cycle,delabays2017multistability}. The interpretation of this condition is straightforward: If we add up all the differences $\theta_n - \theta_m$ around a cycle, the sum must equal zero or an integer multiple of $2\pi$. Notably, if this condition is satisfied for the fundamental cycles, it is satisfied for all cycles. 

The quantities $z_\alpha$ are referred to as winding numbers, and we summarize them in the winding vector $\vec z = \left( z_1, \ldots, z_{\nl-\nn+1} \right)^\top$. These winding numbers are particularly useful to characterize the solutions of the real power flow equations as they are unique \cite[Theorem 4.1]{jafarpour2022flow}. That is, each normal solution corresponds to one unique value of the winding vector $\vec z \in \mathbb{Z}^{\nl-\nn+1}$. Moreover, the number of possible winding vectors is finite. From Eq.~\eqref{eq:cycle-condition} we see that solutions can only be found if
\begin{align}
    |z_{\alpha}| \le \frac{\mbox{number of edges in cycle } \, \alpha}{4} \, .
    \label{eq:z-max}
\end{align}
Generalizations of this approach to power grids with Ohmic losses have been proposed in \cite{balestra2019multistability,delabays2022multistability}.

A summary of the symbols and variables used in this article is provided in Table~\ref{tab:notation} to improve the readability.

\begin{table}[tb]
\caption{
List of symbols and  variables. Vectors are written as boldface lowercase Roman letters and matrices as boldface uppercase case Roman letters. Calligraphic letters are used to denote objective and Lagrangian functions, while Gothic-type letters denote sets and graphs.
}
\label{tab:notation}
\begin{tabular}{p{1.7cm} p{6.3cm}}
\hline
   $e,a,b$ & indices of lines/edges \\
   $\alpha,\beta, \varphi$ & indices of (fundamental) cycles \\
   $C_{a,\beta}$ & entry of the edge cycle incidence matrix \\
   $\matr{C}$ & the edge cycle incidence matrix \\
   $\EE$ & set of all lines or edges in the grid \\
   $E_{ne}$ & entry of the node edge incidence matrix \\
   $\matr{E}$ & the node edge incidence matrix \\
   $f_e$ & real power flow on line $e$ \\
   $\OO$ & objective function \\
   $\vec f$ & vector of all real-power flows \\
   $\GG$ & a graph or network \\
   $K_e$ & coupling strength of a line $e$ (proportional to the inverse of the reactance $X_e^{-1}$) \\
   $\matr K$ & diagonal matrices of coupling constants \\
   $l,m,n$ & indices of nodes \\
   $\matr{L}$ & Laplacian matrix \\
   $\LL$ & Lagrange function \\
   $\ell_\beta$ & loop flow amplitude on the cyle $\beta$ \\
   $\vec{\ell}$ & vector of all loop flow amplitudes\\
   $p_n$ & real power injection at node $n$ \\
   $\vec p$ & vector of all real-power injections \\
   $\theta_n$ & voltage phase angle at node $n$ \\
   $\VV$ & set of all nodes or vertices in the grid\\
  \hline
\end{tabular}
\end{table}

\section{Real power flow from optimization}
\label{sec:powerflow-opt}

In this section, we introduce a convex optimization-based formulation of the nonlinear real power flow equations
\begin{equation}
    p_n = \sum_{m \in \VV} K_{nm} \sin(\theta_n - \theta_m), \qquad
    \forall m \in \VV.
    \label{eq:realpower-section}
\end{equation}
We will demonstrate that this approach allows to compute \emph{all} normal solution, and we will discuss several approximation schemes.
Notably, our approach focuses on the lines and cycles of the grid instead of the nodes.

\subsection{Solutions by convex optimization}

We consider the optimization problem
\begin{align}
    \label{opt:realpower1}
    & \min_{\vec f} \OO_{\rm rp}(\vec f) \\
    & \mathrm{s.t.}  -K_e \le f_e \le K_e, 
         \qquad \forall \, e \in \EE \label{opt:polytope} \\
    & \qquad p_n = \sum_{e} E_{ne} f_e 
        \qquad \forall \, n \in \VV \label{opt:affine}
\end{align}
with the objective function
\begin{align}
    \OO_{\rm rp}(\vec f) = \sum_{e \in \EE} f_e \arcsin \left( \frac{f_e}{K_e} \right) + \sqrt{K_e^2 - f_e^2} - K_e.
\end{align}

\begin{lem}
The optimization problem \eqref{opt:realpower1} is either infeasible or it has a unique solution. 
\end{lem}

\begin{proof}
The feasible space is the intersection of a polytope (defined by Eq.~\eqref{opt:polytope}) and an affine space (defined by Eq.~\eqref{opt:affine}) and, therefore, convex.
The objective function is strictly  convex. We show this statement  by computing its Hesse matrix,
\begin{align*}
    \frac{\partial^2 \OO_{\rm rp}}{\partial f_a \partial f_b}
    = \left\{ \begin{array}{l l l} 
      (K_a^2 - f_a^2)^{-1/2}  & \; \mbox{for} \; &
      a=b \\
      0 & & a \neq b.
    \end{array} \right.
\end{align*}
Hence the Hesse matrix is positive definite on the interior of the feasible space. 
We conclude that the optimization problem is either infeasible or it has a unique solution. 
\end{proof}

\begin{thm}
\label{thm:realpower1}
If the optimization problem \eqref{opt:realpower1} is primarily feasible and the solution satisfies 
\begin{align*}
    |f_e| < K_e \qquad \forall \, e \in \EE,
\end{align*}
then the solution coincides with a normal solution of the real power flow equations \eqref{eq:realpower-section} where the phases $\theta_n$ correspond to the Lagrangian multipliers of the problem.
\end{thm}

\begin{proof}
We solve the optimization problem by introducing the Lagrangian function
\begin{align*}
    \LL(\vec f) =& \OO_{\rm rp}(\vec f) 
     + \sum_{e \in \EE} \mu_e (f_e-K_e) 
     + \sum_{e \in \EE} \nu_e (-f_e-K_e) \\
     & + \sum_{n \in \VV} \lambda_n \left( p_n - 
     \sum_{e \in \EE} E_{ne} f_e \right).
\end{align*}
The optimization problem satisfies the Slater condition, hence every optimum fulfills the KKT conditions:
\begin{enumerate}
    \item Stationarity with respect to $f_e$:
    \begin{align*}
        \frac{\partial \LL}{\partial f_e}
        = \arcsin \left( \frac{f_e}{K_e} \right) + (\mu_e - \nu_e) - \sum_n \lambda_n E_{ne} = 0
    \end{align*}
    for all edges $e \in \EE$,
    \item Primal feasibility:
    \begin{align*}
    &-K_e \le f_e \le K_e, 
         \qquad \forall \, e \in \EE \nonumber \\
    & p_n = \sum_{e} E_{ne} f_e 
        \qquad \forall \, n \in \VV \nonumber,
    \end{align*}
    \item Dual Feasibility:
    \begin{align*}
        \mu_e \ge 0, \; \nu_e \ge 0 
        \qquad \forall \, e \in \EE, 
    \end{align*}
    \item Complementary slackness:
    \begin{align*}
        \mu_e (f_e-K_e) =0, \; \nu_e (-f_e-K_e) = 0
        \qquad \forall \, e \in \EE.
    \end{align*}
\end{enumerate}
By assumption we have $|f_e| < K_e$ such that complementary slackness condition yields $\mu_e = \nu_e=0$ for all edges $e \in \EE$. The stationarity condition with respect to the variables $f_e$ then reads
\begin{align*}
    & \arcsin \left( \frac{f_e}{K_e} \right) = \sum_n E_{ne} \lambda_n \\
    \Leftrightarrow \quad &
    f_e = K_e \sin\left( \sum_n E_{ne} \lambda_n \right),
\end{align*}
which we substitute into the equality constraint
\begin{align*}
    p_m = \sum_{e \in \EE} E_{me} K_e \sin\left( \sum_n E_{ne} \lambda_n \right).
\end{align*}
Now we switch our notation and label the edges by their respective endpoints. Using the structure of the node-edge incidence matrix, the last equation can be recast into the form 
\begin{align*}
    p_m = \sum_{n \in \VV}  K_{nm} \sin\left( \lambda_n - \lambda_m \right).
\end{align*}
This coincides with the real power flow equations \eqref{eq:realpower-section} if we identify $\lambda_n$ and $\theta_n$.
\end{proof}

Theorem \ref{thm:realpower1} provides a systematical approach to the real power flow equations and the stationary states of oscillator networks. The optimization problem \eqref{opt:realpower1} either has a unique solution or no solution at all. However, we know that the Kuramoto model is generally multistable, i.e., it can have multiple stable fixed points. Can we get them all by convex optimization? 

As stressed above we will restrict ourselves to normal operation fixed points for the time being. In Sec.~\ref{sec:notation-review} we have shown that all solutions of the real power flow equations can be obtained from the general solution of the KCL
\begin{align*}
    \vec f = \vec f^{(0)} + \matr C \, \vec \ell.
\end{align*}
Here, $\vec f^{(0)}$ denotes a special solution of the KCL -- for instance the one that is obtained from solving the optimization problem \eqref{opt:realpower1}. Then the cycle flow amplitudes $\vec \ell$ must satisfy the conditions 
\begin{align}
    \sum_e C_{e \beta} \arcsin\left( \frac{f_e}{K_e} \right) = 2 \pi z_\beta ,
    \label{eq:cyclecon2}
\end{align}
for all fundamental cycles $\beta$. Notably, the winding vector $\vec z = (z_1, \ldots z_{M-N+1})$ is unique for every fixed point as discussed above. Hence, we can actually try to define a specific convex optimization problem for a given $\vec z$, whose unique solution then reproduces the condition \eqref{eq:cyclecon2}. Indeed, this is possible as we will show now.

Consider the optimization problem
\begin{align}
    \label{opt:realpower2-cycle}
    & \min_{\vec \ell} \OO_{\vec z}(\vec \ell) \\
    & s.t.  -K_e \le f^{(0)}_{e} + \sum_\beta  C_{e \beta} \ell_\beta \le K_e, 
         \qquad \forall \, e \in \EE \nonumber
\end{align}
with the objective function
\begin{align}
    &\OO_{\vec z}(\vec \ell) =
     \sum_{e \in \EE}  \sqrt{K_e^2 - (f^{(0)}_{e} + \sum_\beta  C_{e \beta} \ell_\beta )^2} - K_e \nonumber \\
    & + \sum_{e \in \EE} 
    \left( f^{(0)}_{e} + \sum_\beta  C_{e \beta} \ell_\beta   \right)
    \arcsin \left( \frac{f^{(0)}_{e} + \sum_\beta  C_{e \beta} \ell_\beta}{K_e} \right) \nonumber \\ 
    & - \sum_{\beta}2 \pi z_\beta \ell_\beta .
    \label{eq:objective-z}
\end{align}

\begin{lem}
The optimization problem \eqref{opt:realpower2-cycle} is either infeasible or it has a unique solution. 
\end{lem}

\begin{proof}
The feasible space is the intersection of a polytope and an affine space and, therefore, convex (or empty).
The objective function is strictly convex. To prove this statement, we compute the Hesse matrix
\begin{align*}
    \frac{\partial \OO_{\vec z}}{\partial \ell_\alpha}
    &=  \sum_e C_{e \alpha} \arcsin \left( \frac{f^{(0)}_{e} + \sum_{\varphi} C_{e \varphi} \ell_\varphi}{K_e} \right)
    - 2 \pi z_\alpha \\
    \frac{\partial^2 \OO_{\vec z}}{\partial \ell_\alpha \partial \ell_\beta}
    &= \sum_e \frac{C_{e \alpha} C_{e \beta}}{\sqrt{K_e^2 - (f^{(0)}_{e} + \sum_\varphi  C_{e \varphi} \ell_\varphi )^2}} \, .
\end{align*}
Converting this result to a matrix form, the Hesse matrix reads
\begin{align}
    \nabla^2 \OO_{\vec z} = \matr{C}^\top
    \mbox{diag} \left( \frac{1}{\sqrt{K_e^2 - (f^{(0)}_{e} + \sum_\varphi  C_{e \varphi} \ell_\varphi )^2}}
    \right) \matr{C} \, .
    \label{eqn:Hesse-ell}
\end{align}
We see that the Hesse matrix is positive definite on the interior of the feasible set where 
\begin{align*}
    \left| f^{(0)}_{e} + \sum_\beta  C_{e \beta} \ell_\beta \right| < K_e \, .
\end{align*}
whereas it becomes singular on the boundary. We conclude that the optimization problem is either infeasible or has a unique solution.
\end{proof}

\begin{thm}
\label{thm:realpower2-cycle}
If the optimization problem \eqref{opt:realpower2-cycle} is primarily feasible and the solution satisfies 
\begin{align*}
    |f^{(0)}_{e} + \sum_\beta  C_{e \beta} \ell_\beta| < K_e \qquad \forall \, e \in \EE,
\end{align*}
then the solution coincides with a normal solution of the real power flow equations \eqref{eq:realpower-section} with the winding vector $\vec z = (z_1, \ldots , z_{M-N+1})$.
\end{thm}

\begin{proof}
We solve the optimization problem by introducing the Lagrangian function
\begin{align*}
    \LL(\vec \ell) =& \OO_{\vec z}(\vec \ell) \\
     & + \sum_{e \in \EE} \mu_e \left( f^{(0)}_{e} + \sum_\beta  C_{e \beta} \ell_\beta  -K_e \right) \\
     & + \sum_{e \in \EE} \nu_e \left( -f^{(0)}_{e} - \sum_\beta  C_{e \beta} \ell_\beta - K_e \right).
\end{align*}
The optimization problem satisfies the Slater condition, hence every optimum fulfills the KKT conditions:
\begin{enumerate}
    \item Stationarity with respect to $\ell_\alpha$:
    \begin{align*}
        & 0 = \frac{\partial \LL}{\partial \ell_\alpha} = - 2 \pi z_\alpha \\        
        & +\sum_e C_{e \alpha} \left[ \arcsin \left( \frac{f^{(0)}_{e} + \sum_\beta  C_{e \beta} \ell_\beta }{K_e} \right) + (\mu_e - \nu_e) \right].
    \end{align*}
    for all fundamental cycles $\beta$,
    \item Primal feasibility:
    \begin{align*}
          -K_e \le f^{(0)}_{e} + \sum_\beta  C_{e \beta} \ell_\beta \le K_e
    \end{align*}
    \item Dual Feasibility:
    \begin{align*}
        \mu_e \ge 0, \; \nu_e \ge 0 
        \qquad \forall \, e \in \EE, 
    \end{align*}
    \item Complementary slackness:
    \begin{align*}
        &\mu_e \left( f^{(0)}_{e} + \sum_\beta  C_{e \beta} \ell_\beta  -K_e \right) =0, \\ 
        &\nu_e \left( - f^{(0)}_{e} - \sum_\beta  C_{e \beta} \ell_\beta - K_e \right) = 0,
        \qquad \forall \, e \in \EE.
    \end{align*}
\end{enumerate}
By assumption we have $|f^{(0)}_{e} + \sum_\beta  C_{e \beta} \ell_\beta| < K_e$ for all $e \in \EE$ such that complementary slackness condition yields $\mu_e = \nu_e=0$ for all edges $e \in \EE$. The stationarity condition with respect to the variables $\ell_\alpha$ then reads
\begin{equation}
    \sum_{e=1}^{\nl} C_{e,\alpha} \arcsin\left(  \frac{f_e^{(0)}  + \sum_{\beta=1}^{\nl-\nn+1} 
    C_{e, \beta}
    \ell_\beta }{K_e} \right) = 2 \pi z_\alpha \, .
    \label{eq:stat-vs-cycle}
\end{equation}
Hence, the solution of the optimization problem satisfies the two conditions for the  normal solution of the real power flow equations given in Sec.~\ref{sec:notation-review}. First, it satisfies the KCL by construction. Second, the stationarity condition \eqref{eq:stat-vs-cycle} coincides with the cycle condition \eqref{eq:cycle-condition}.
\end{proof}

We close this section with three remarks on the implications of the two theorems established above. First, the optimization problem \eqref{opt:realpower1} yields a solution with winding vector $\vec z = \vec 0$ (or no suitable solution at all). To see this, use this solution as the reference $\vec f^{(0)}$ in the optimization problem \eqref{opt:realpower2-cycle} with $\vec z =  \vec 0$. Then the minimizer is simply given by $\vec \ell = \vec 0$. Hence, the solution $\vec f^\star = \vec f^{(0)} + \matr C \vec \ell = \vec f^{(0)}$ has the winding number $\vec z = \vec 0$. Normal solutions with $\vec z=0$ are by far the most important ones in practice -- all other ones correspond to rather exotic states with loop flows, see \cite{coletta2016topologically} for a discussion. 

Second, we emphasize that theorem \ref{thm:realpower2-cycle} provides a systematic approach to compute \emph{all} normal solutions of the real power flow equations \eqref{eq:realpower-section}. More precisely, we can find the following scenarios:
\begin{enumerate}
    \item Given the network and the real power injections $\vec p$, the optimization problems may be infeasible. That is, there is no solution of the KCL $\matr E \vec f = \vec p$ that satisfies the line limits $|f_e|\le K_e$. In this case, the grid simply does not have enough capacity to transmit the power from generators to consumers. Whether this is the case can be determined in a systematic way by graph-theoretical methods, see appendix \ref{app:kcl}.
    \item If there is a solution $\vec f^{(0)}$ of the KCL satisfying the line limits, then the optimization problem \eqref{opt:realpower2-cycle} has a unique solution for each winding vector $\vec z$. As the winding vector is unique, we can hence systematically compute the respective solution and also decide whether it exists at all by distinguishing three cases:
    \begin{enumerate}
        \item If the solution is in the interior of the feasible set (i.e., $|f_e| < K_e$ for all lines $e$), then we have found the correct solution for the given $\vec z$. 
        \item If the solution lies on the boundary of the feasible space (i.e., $|f_e| = K_e$ for at least one line $e$), then generally no solution with the given winding vector $\vec z$ exists. 
        \item In special case, we may find a solution with $|f_e| = K_e$ where the respective KKT multiplier ($\mu_e$ or $\nu_e$) vanishes, too. Then we have found a valid solution of the real power flow equations. This case typically corresponds to a bifurcation point, where the solution will vanish upon the variation of a system parameter.
    \end{enumerate}
\end{enumerate}

Finally, we remark that the suggested procedure to compute all normal solutions may still be computationally hard, depending on the topology of the network. For every fundamental cycle $\alpha$, the number of allowed values $z_\alpha$ is finite according to Eq.~\eqref{eq:z-max}. The number of allowed winding vectors $\vec z$ is thus also finite, but it can grow exponentially with the number of cycles.

\section{The linear power flow approximation}
\label{sec:DCapprox}

The linear power flow or DC approximation is widely used in practical approximations. Here, one simply linearizes the sine function in Eq.~\eqref{eq:realpower-section} and obtains a system of linear equations
\begin{equation}
    p_n = \sum_{m \in \VV} K_{nm} (\theta_n - \theta_m), \qquad
    \forall m \in \VV.
    \label{eq:dcapprox-section}
\end{equation}
which is easily solved for the nodal phase angles as discussed in Sec.~\ref{sec:notation-review}. Here we discuss how this approximation relates to our optimization approach and how this approximation may be refined.

\subsection{The linear power flow as an optimization problem}

We first note that the linear power flow equations can also be obtained from a convex optimization problem:
\begin{align}
    \label{opt:dcapprox1}
    & \min_{\vec f} \OO_{\rm lin}(\vec f) \\
    & \qquad s.t.~p_n = \sum_{e} E_{ne} f_e 
        \qquad \forall \, n \in \VV \nonumber
\end{align}
with the objective function
\begin{align}
    \OO_{\rm lin}(\vec f) = \sum_{e \in \EE} \frac{f_e^2}{2 K_e} \, .
\end{align}
To see this, we introduce the Lagrangian
\begin{align*}
    \LL(\vec f) =& \OO_{\rm lin}(\vec f) 
     + \sum_{n \in \VV} \lambda_n \left( p_n - 
     \sum_{e \in \EE} E_{ne} f_e \right).
\end{align*}
The stationarity condition with respect to $f_e$ then reads:
\begin{align*}
    \frac{\partial \LL}{\partial f_e}
    = \frac{f_e}{K_e} - \sum_n \lambda_n E_{ne} = 0
\end{align*}
for all edges $e \in \EE$. Substituting this result into the equality constraints yields
\begin{align*}
    p_m = \sum_{e \in \EE} E_{me} K_e \left( \sum_n E_{ne} \lambda_n \right).
\end{align*}
Switching our notation and labeling the edges by their respective endpoints, we can recast this equation into the form 
\begin{align*}
    p_m = \sum_{n \in \VV}  K_{nm} \left( \lambda_n - \lambda_m \right).
\end{align*}
This coincides with the linear power flow equations \eqref{eq:dcapprox-section} if we identify $\lambda_n$ and $\theta_n$.

Now how does the optimization problem \eqref{opt:dcapprox1} relate to the previous problem \eqref{opt:realpower1}? By straightforward computation, one can show that the objective function $\OO_{\rm lin}$ is just the leading order Taylor expansion of the original objective function $\OO_{\rm rp}$ around the ``empty grid'' $\vec f = \vec 0$:
\begin{align*}
    \OO_{\rm rp}(\vec f)= \OO_{\rm lin}(\vec f) + \mathcal{O}(f_e^4).
\end{align*}
Furthermore, the optimization problem \eqref{opt:dcapprox1} neglects the inequality constraints, i.e.,~the line limits.

\subsection{Improving on the linear power flow approximation}
\label{sec:improving-dc-newton}

The previous insights provide a method to improve the linear power flow approximation. In the first step we linearize the objective $\OO_{\rm rp}(\vec f)$ around $\vec 0$ and obtain the linear power flow $\vec f^{\rm lin}$. We can repeat this idea and do another Taylor expansion, but this time around the previous approximate solution $\vec f^{\rm lin}$. 

Let's evaluate this idea. The linear power flow solution is given by \begin{align}
    \vec \theta^{\rm(lin)} &= \matr L^{+} \vec p, \nonumber \\
    \vec f^{\rm(lin)}  &= \matr K \matr E^\top \vec \theta^{\rm(lin)} \, .
    \label{eq:dcapprox-vec}
\end{align}
For a given edge $e = (n,m)$, we further define $\theta_e^{\rm(lin)} = \theta^{\rm(lin)}_n - \theta^{\rm(lin)}_m$.
Furthermore, we use the Taylor expansion of the objective function. Writing $f_e = f_e^{\rm(lin)} + \Delta f_e$, we obtain 
\begin{align*}
    &f_e \arcsin\left( \frac{f_e}{K_e} \right) + \sqrt{K_e^2 - f_e^2} -K_e \\
    &= f^{\rm (lin)}_e \arcsin\left( \frac{f^{\rm (lin)}_e}{K_e} \right) 
    + \sqrt{K_e^2 - f^{\rm (lin)2}_e} -K_e \\
    &+ \arcsin \left( \frac{f^{\rm (lin)}_e}{K_e} \right) \, \Delta f_e
    + \frac{1}{\sqrt{ K_e^2 - f^{\rm (lin)2}_e }} \, \Delta f_e^2
    + \mathcal{O}(\Delta f_e^3).
\end{align*}
Now we can rewrite the optimization problem \eqref{opt:realpower2-cycle}. Setting $\Delta \vec f = \matr C \vec \ell$ and $\vec z = \vec 0$, the objective function reads
\begin{align*}
     \OO_{\vec z}(\vec \ell) &=
    \sum_{e \in \EE}  f^{\rm (lin)}_e \arcsin\left( \frac{f^{\rm (lin)}_e}{K_e} \right) 
    + \sqrt{K_e^2 - f^{\rm (lin)2}_e} \\
    & - K_e +\sum_{e \in \EE} 
    \arcsin \left( \frac{f^{\rm (lin)}_e}{K_e} \right) \, 
    \sum_{\beta} C_{e \beta} \ell_\beta  \\
    & + \sum_{e\in \EE} \frac{1}{\sqrt{ K_e^2 - f^{\rm (lin)2}_e }} \, \left( \sum_{\beta} C_{e \beta} \ell_\beta \right)^2 \nonumber 
    + \mathcal{O}(\ell^3).
\end{align*}
We ignore the inequality constraints for the time being, as in the linear power flow approximation, and discard the higher-order terms. The minimizer is then found by requiring the objective to be stationary:
\begin{align*}
    0 = \frac{\partial \OO_{\vec z} }{\partial \ell_\alpha} = & \sum_{e} \arcsin \left( \frac{f^{\rm (lin)}_e}{K_e} \right) C_{e \alpha} \\
    &+ \sum_{e, \beta} \frac{1}{\sqrt{ K_e^2 - f^{\rm (lin)2}_e }} C_{e \alpha} C_{e \beta} \ell_\beta \, .
\end{align*}
We rewrite this set of linear equations in a vectorial notation defining the diagonal matrix $\matr K^{\rm red} \in \mathbb{R}^{M \times M}$ with elements
\begin{align*}
    \matr K^{\rm  red}_{ee} =
    \left[ K_e^2 - f^{\rm (lin)2}_e \right]^{-1/2} \, .
\end{align*}
We thus obtain the equation
\begin{align*}
    \left( \matr C^\top \matr K^{\rm red} \matr C \right)
    \vec \ell
    = - \matr C^\top \arcsin \left(  \matr K^{-1} \vec f^{\rm (lin)}\right)  , 
\end{align*}
where the arcsin is taken element-wise, which is readily solved for $\vec \ell$. We thus obtain an approximate solution of the nonlinear real power flow equation,
\begin{align}
\begin{split}
\label{eq:improved_lin_approx}
\vec f^{\rm (approx.)} &= \vec f^{\rm (lin)} + \matr C  \vec \ell \\
    &= \vec f^{\rm (lin)} - \matr C 
    \left[ \matr C^\top \matr K^{\rm red} \matr C \right]^{-1}
    \matr C^\top \\
    & \qquad \qquad  \cdot \arcsin \left( \matr K^{-1} \vec f^{\rm (lin)} \right)  
    \, .
\end{split}    
\end{align}
Notably, this expression is no longer linear in the power injections $\vec p$. Nevertheless, it provides an explicit formula for the flows $\vec f^{\rm (approx.)}$ in closed form. 

To evaluate the quality of our improved approximation of the real power flow, we compare the errors of the purely linear approximation with the errors of our improved approximation for a simple 30-bus test case, see Fig.~\ref{fig:case_30_improving_lin_newton}. In this calculation, we have used the minimal cycle base of the network.
For the heavily loaded test grid, our approximation improves the median of the approximation errors by at least three orders of magnitude and the error on the heaviest loaded line by two magnitudes. For a grid with small loads, the improvement is up to nine orders of magnitude.

Since the correction term can be calculated purely algebraically and, in particular, improves the errors on the heavily loaded lines, this improved approximation can be of interest for any application where running a full load flow consumes too many resources, e.g.,~in contingency analysis.

\begin{figure}
    \centering
    \includegraphics[width=\columnwidth]{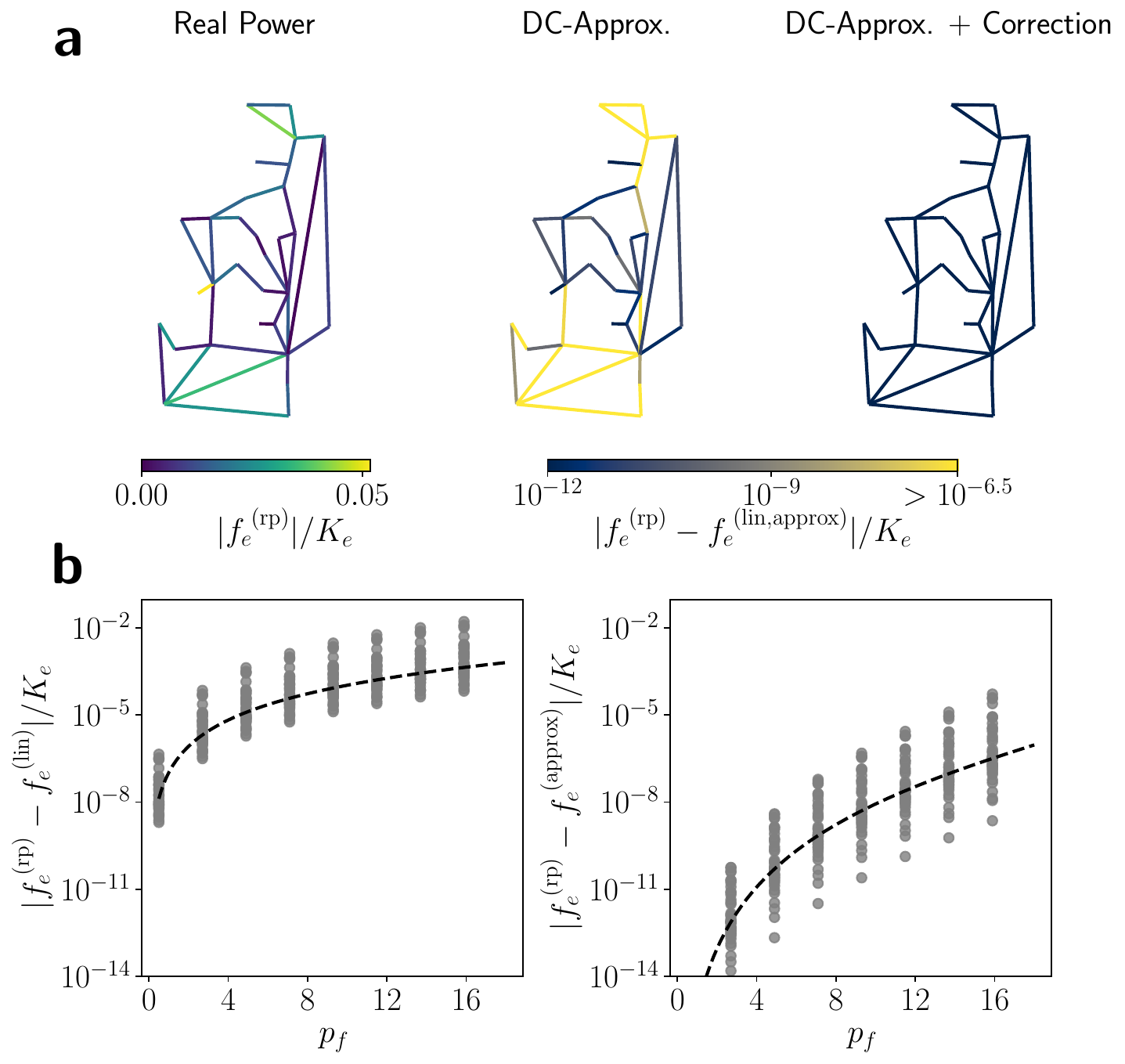}
    \caption{
    Assessment of the improved power flow 
    approximation~\eqref{eq:improved_lin_approx}.
    We compute the real power flows $\vec f$ for an adapted \textsc{Matpower} 30-Bus test case, cf.~Appendix~\ref{sec:app:test_case}, and compare numerically exact values $\vec f^{\rm (rp)}$ to the linear (DC) approximation $\vec f^{\rm (lin)}$ and the improved approximation $\vec f^{\rm (approx)}$ given by Eq.~\eqref{eq:improved_lin_approx}.
    (a) The left panel shows the loading  $|f^{\rm (rp)}_e | / K_e$ for each edge $e$ for the exact solution. 
    The other two panels depict the errors $|f^{\rm (rp)}_e - f^{\rm (lin, approx)}_e| /  K_e$ of the linear (middle) and the improved approximation (right). Errors are hardly visible for the improved 
    approximation~\eqref{eq:improved_lin_approx}, 
    even on a logarithmic scale
    (b) For a systematic assessment we increase the overall grid load by multiplying all power injections by a scaling factor $p_f$. The figure shows the error on all edges as a function of $p_f$ in a scatter plot for the linear (left) and the improved (right) approximation. The improved 
    approximation~\eqref{eq:improved_lin_approx} reduces the error on most edges by at least two orders of magnitudes for $p_f = 20$ and up to nine orders for $p_f=1$.
    For $p_f=1$, the errors of the improved approximation are below $10^{-14}$ approaching the numerical precision. 
    \label{fig:case_30_improving_lin_newton}
    }
\end{figure}

\subsection{Bounding the error in the linear power flow approximation}

We may use our insights to bound the error introduced by the linear power flow approximation. That is, we derive a bound on the norm of the difference of real and linear power solutions,
\begin{align}
    \vec \xi = \vec f^{\rm (rp)} -\vec f^{\rm (lin)}.
\end{align}

To formulate the bound, we define the function
\begin{align}
    \GC(\vec f) := \OO_{\rm rp}(\vec f) - \OO_{\rm lin}(\vec f).
\end{align}
Decomposing this function into components $\GC(\vec f) = \sum_e \GC_e(f_e)$, we have
\begin{align*}
    \GC_e(f_e) &= f_e \arcsin \left( \frac{f_e}{K_e} \right) + \sqrt{K_e^2 - f_e^2} - K_e - \frac{f_e^2}{2 K_e} ,\\
    \GC'_e(f_e) &= \arcsin \left( \frac{f_e}{K_e} \right) - \frac{f_e}{K_e} \, ,
\end{align*}
where the prime denotes the derivative with respect to the argument. One can see that $\GC_e(f_e) \ll 1$ if the loads are small $f_e \ll K_e$. Furthermore, we define the vector $\vec \zeta$ with components
\begin{equation}
    \zeta_e = K_e \GC'_e(f_e^{\rm (lin)}).
\end{equation}
We then obtain the following bound for the deviation of real and linear power flow.

\begin{thm}
\label{thm:david1}
Let $\vec f^{\rm (rp)}$ denote the solution of the real power equations with $\vec z= \vec 0$ and $\vec f^{\rm (lin)}$ the solution of the linear power flow equations, assuming that both satisfy $|f_e|< K_e$. Then the difference
\begin{align}
    \vec \xi = \vec f^{\rm (rp)} -\vec f^{\rm (lin)}
\end{align}   
is bounded as
\begin{align}
    \| \vec \xi \|_K & \le  \| \vec \zeta \|_K =   
     \left[ \sum_e K_e \GC'_e(f_e^{\rm (lin)})^2 \right]^{1/2}, 
    \label{eq:bound-xi-1} \\
    \| \vec \xi \|_K & \le  \| \matr \Pi_{\rm cycle} \vec \zeta \|_K,
    \label{eq:bound-xi-2}
\end{align}  
where we used the inner product \eqref{eq:def-inner-K} and the associated norm
$\| \vec \xi \|_K^2 :=  \langle \vec \xi, \vec \xi \rangle_K$.
\end{thm}

\begin{proof}
We have
\begin{align*}
    \OO_{\rm rp}(\vec f^{\rm (rp)}) = 
    \OO_{\rm lin}(\vec f^{\rm (lin)} + \vec \xi) + \GC(\vec f^{\rm (lin)} + \vec \xi)
\end{align*}
We can rewrite the first term as
\begin{align*}
    \OO_{\rm lin}(\vec f^{\rm (lin)} + \vec \xi) 
    &= \frac{1}{2} \| \vec f^{\rm (lin)} + \vec \xi \|_K^2 \\
    &= \OO_{\rm lin}(\vec f^{\rm (lin)})
        + \frac{1}{2} \| \vec \xi \|_K^2
\end{align*}
since $\langle \vec \xi, \vec f^{\rm (lin)} \rangle_K = 0$ for all $\vec \xi$ in the kernel of $\matr E$. 
Using the first-order Taylor expansion of the function $\GC_e$ and the fact that the function is convex, we can bound the second term from below,
\begin{align*}
     \GC(\vec f^{\rm (lin)} + \vec \xi)
     \ge \sum_e \GC_e(f_e^{\rm (lin)}) + \GC'_e(f_e^{\rm (lin)})  \, \xi_e.
\end{align*}
We thus obtain
\begin{align}
    \OO_{\rm rp}(\vec f^{\rm (rp)}) \ge
    \OO_{\rm rp}(\vec f^{\rm (lin)}) 
    + \frac{1}{2} \| \vec \xi \|_K^2
    +  \langle \vec \zeta , \vec \xi \rangle_K \, .
    \label{eq:proofd-inequality1}
\end{align}
On the other hand, the vector $\vec f^{\rm (rp)}$ is the minimizer of the function $\OO_{\rm rp}$ by definition such that $ \OO_{\rm rp}(\vec f^{\rm (lin)}) \ge \OO_{\rm rp}(\vec f^{\rm (rp)})$. Applying Taylor's theorem as detailed in appendix \ref{app:taylor}, we can obtain the improved bound
\begin{align}
    \OO_{\rm rp}(\vec f^{\rm (lin)})
    \ge
    \OO_{\rm rp}(\vec f^{\rm (rp)})
    + \frac{1}{2} \| \vec \xi \|_K^2 
    \label{eq:Frpfrp-smaller-Frpflin}
\end{align}
Combining with the inequality \eqref{eq:proofd-inequality1}, we conclude that 
\begin{align}\label{eqn:can be sharpened}
     \| \vec \xi \|_K^2
    +  \langle \vec \zeta , \vec \xi \rangle_K \le 0.
\end{align}
Applying the Cauchy-Schwarz inequality then yields Eq.~\eqref{eq:bound-xi-1}.

To obtain the second bound we note that $\vec \xi$ is a pure cycle flow because both $\vec f^{\rm (lin)}$ and $\vec f^{\rm (rp)}$ satisfy the KCL. Hence we find that $\vec \xi = \matr \Pi_{\rm cycle} \vec \xi$. Using the orthogonality of the projection with respect to the inner product \eqref{eq:def-inner-K}, we can rewrite the inequality \eqref{eqn:can be sharpened} and obtain
\begin{align*}
    0 \; \ge \;  & \frac{1}{2} \| \vec \xi \|_K^2
    +  \langle \vec \zeta , \vec \xi \rangle_K \\
    & = \frac{1}{2} \| \vec \xi \|_K^2
    +  \langle \vec \zeta , \matr \Pi_{\rm cycle} \vec \xi \rangle_K \\
    & = \frac{1}{2} \| \vec \xi \|_K^2
    +  \langle \matr \Pi_{\rm cycle}  \vec \zeta , \vec \xi \rangle_K \, .
\end{align*}
Applying the Cauchy-Schwarz inequality and squaring the result then yields Eq.~\eqref{eq:bound-xi-2}.  
\end{proof}

\begin{figure}[tb]
\includegraphics[width=\columnwidth]{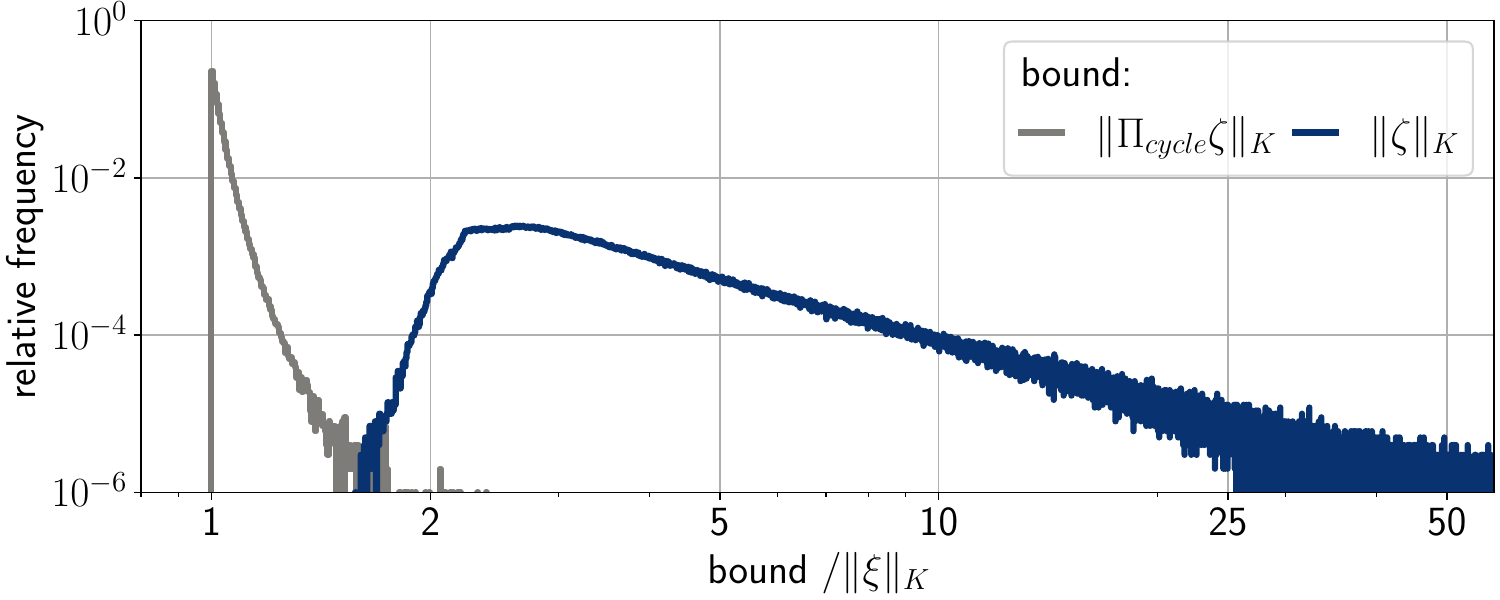}
\caption{
\label{fig:bounds}
Assessment of bounds of theorem~\ref{thm:david1} for the error of the linear power flow approximation.
We solve the real power flow and the linear power flow equations for the adapted \textsc{Matpower} 30-bus test case, cf.~appendix \ref{sec:app:test_case}, and $10^6$ randomly sampled and valid power injection vectors $\vec p$.
In each case we compute the norm of error $\| \vec \xi \|_K $  and compare it to the upper bounds given by 
$ \|  \vec \zeta \|_K$ and $ \| \matr \Pi_{\rm cycle} \vec \zeta \|_K$, respectively.
The figure shows a histogram of the ratio $\rm {bound}/{ \| \vec \xi \|_K}$ which serves as a measure for the tightness of the bound. 
The improved upper bound $\| \matr \Pi_{\rm cycle} \vec \zeta \|_K$  is significantly better and appears to be tight.
}
\end{figure}

We numerically test the tightness of the derived bounds. To this end we start from the adapted \textsc{Matpower} 30-bus test case, see Appendix \ref{sec:app:test_case}, choose the power injections $p_n$ at random and numerically compute both the real power flow $\vec f^{\rm (rp)}$ and the linear power flow $\vec f^{\rm (lin)}$. Figure~\ref{fig:bounds} shows that the bound \eqref{eq:bound-xi-2} incorporating the projection is much tighter than the simpler bound \eqref{eq:bound-xi-1}.

Theorem~\ref{thm:david1} provides an upper bound for the $K$-norm of the error $\vec \xi$. We further derive an error bound for each single line, combining the previous result with a general property of cycle flows.

\begin{lem}
\label{lem:cycle_flows_bound_k_norm}
For every cycle flow $\vec f^{(c)}$, i.e., every flow that satisfies $\matr E \vec f^{(c)} = 0$, we have the inequality 
\begin{equation*}
    \| \vec f^{(c)} \|_K^2 \geq \big( f^{(c)}_a \big)^2
    \frac{1}{K_a (1 - K_a \Omega_a)}
\end{equation*}
for every edge $a = (n,m) \in \EE$, where $\Omega_a$ is the effective resistance as defined in Eq.~\eqref{eq:Omega-a}.
\end{lem}

The proof of this lemma is mostly technical such that we postpone it to appendix~\ref{app:proof-norm-max-cycle}. We can apply this inequality with theorem \ref{thm:david1}, taking into account that $\matr E (\vec f^{\rm (rp)}  - \vec f^{\rm (lin)} ) = \vec 0$. One directly obtains the following error bound.

\begin{cor}
\label{lcor:bounds_single_lines}
For every edge $a = (n,m) \in \EE$ we find that 
\begin{equation}
    (f^{\rm (rp)}_a - f^{\rm (lin)}_a)^2 \leq  K_a (1 - K_a \Omega_a) \| \matr \Pi_{\rm cycle} \vec \zeta \|_K^2 \, .
\end{equation}
\end{cor}

\section{Geometry of real and linear power flow}
\label{sec:geometry}

We have shown that both the real and the linear power flow equations can be recast as optimization problems. Comparing the two objective functions $\OO_{\rm rp}(\vec f)$ and $\OO_{\rm lin}(\vec f)$ thus provides insights into the relations of the two problems and the limitations of the linear approximation. Here we put forward a geometric approach to this topic.

\subsection{Geometric interpretation of the cycle condition}

We have discussed in Sec.~\ref{sec:notation-review} that the solutions for the real power flow equations are characterized by two conditions: (i) The continuity equation or KCL \eqref{eq:KCL} and (ii) the cycle condition
(cf.~Eq.~\eqref{eq:cycle-condition})
\begin{equation}
    \sum_{e=1}^{\nl} C_{e,\alpha} \arcsin\left(  \frac{f_e}{K_e} \right) = 2 \pi z_\alpha \, .
    \label{eq:cycle-condition2}
\end{equation}
For the linear flow equations, the analog to the cycle condition is just given by Kirchhoff's voltage law (KVL)
\begin{equation}
    \sum_{e=1}^{\nl} C_{e,\alpha} \left(  \frac{f_e}{K_e} \right) = 0.
    \label{eq:KVL2}
\end{equation}

\begin{figure}[tb]
\centering
\includegraphics[width=\columnwidth]{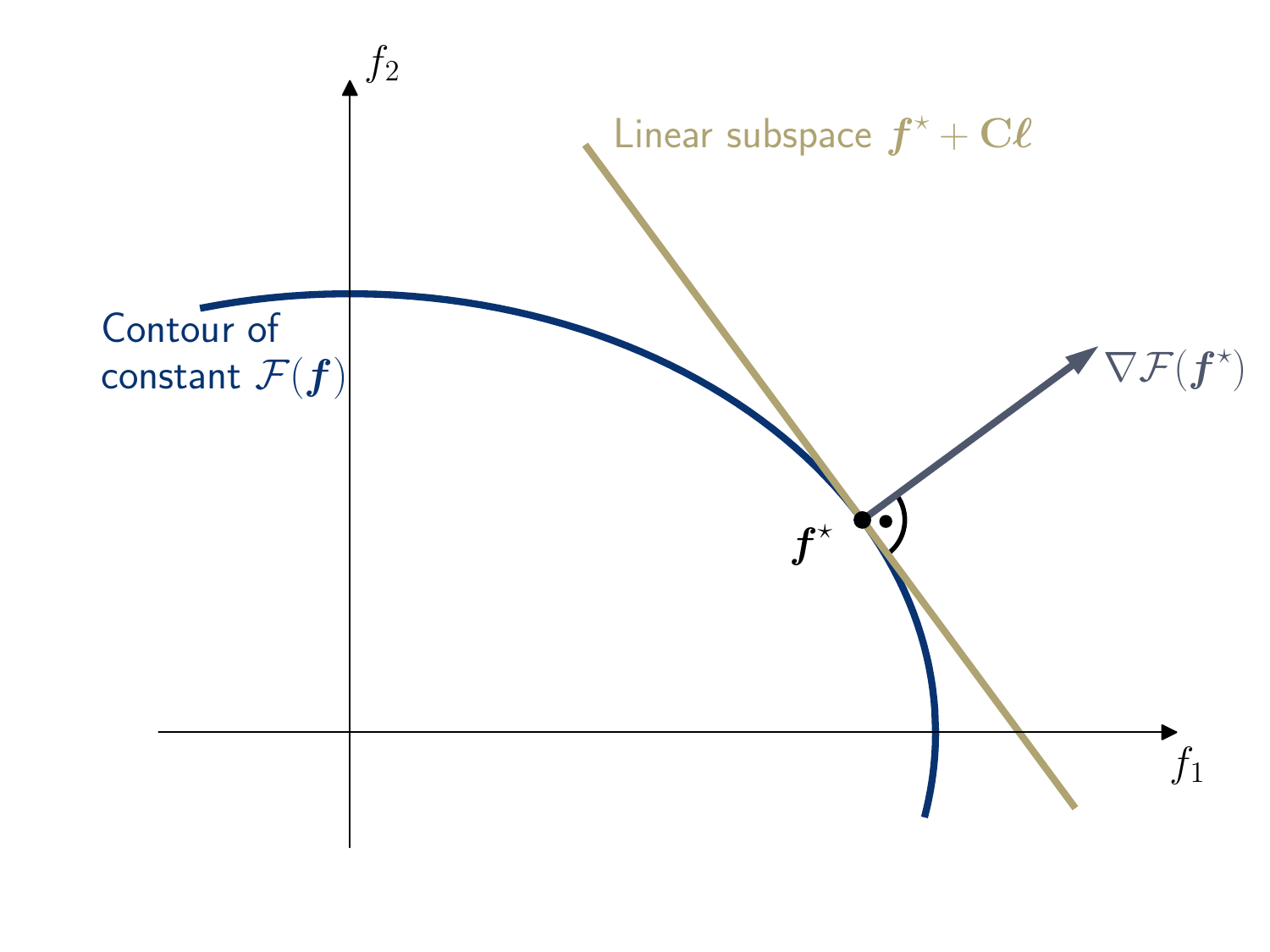}
\caption{Geometric interpretation of the cycle conditions. The optimal solution $\vec f^{\star}$ of the optimization problems~\eqref{opt:realpower1} and~\eqref{opt:dcapprox1} 
are given by the point where contour lines of the respective objective function $\OO (\vec f)$ and the linear subspace spanned by the equality constraints $\matr E \vec f = \vec p$ are tangential.
Equivalently, $\vec f^\star$ is the point where the gradient $\nabla \OO (\vec f)$ is orthogonal to the linear subspace. Since any flow $\vec f$ can be decomposed into directed and a cycle flow $\vec f^{(c)} = \matr C \vec \ell \in \ker (\matr E)$, the linear subspace is spanned by all points of the form $\vec f^{\star} + \matr C \vec \ell$. 
\label{fig:orthogonality}
}
\end{figure}

\begin{figure*}[tb]
\label{fig:grad_descent}
\includegraphics[width=\textwidth]{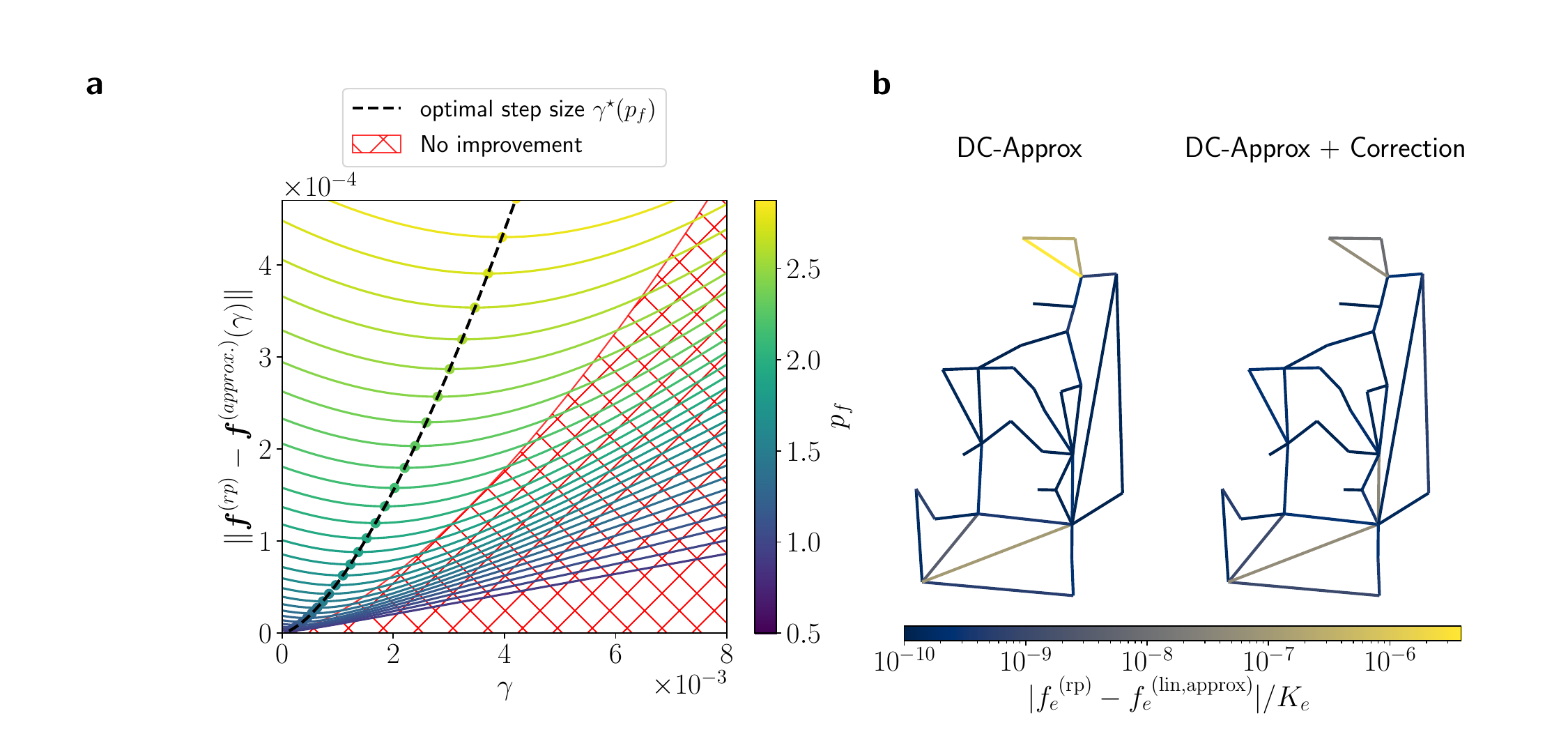}
\caption{Improving the linear power flow approximation using a projected gradient descent.
We compute the real power flows $\vec f$ for an adapted \textsc{Matpower} 30-Bus test case, cf.~appendix \ref{sec:app:test_case}, and compare numerically exact values $\vec f^{\rm (rp)}$ to the improved approximation
given by Eq.~\eqref{eq:gradient-step}.
(a) We show the error of the approximation as a function of the step size $\gamma$ for different values of the scaling factor $p_f$ that controls the grid load.
The optimal step size $\gamma^\star$, drawn as a black dotted line, increases with $p_f$. The red hatched area shows where the gradient descent step does not improve the approximation.
(b) The maps depict the error of the approximation $|f^{\rm (rp)}_e - f^{\rm (lin, approx)}_e| /  K_e$ for each individual line $e$, comparing the linear  approximation and the improved approximation given by Eq.~\eqref{eq:gradient-step} for the optimal step size $\gamma^\star$ and $p_f  = 1.028$.
}
\end{figure*}

We now discuss how the cycle condition and the KVL relate to the optimization problems defined above, restricting ourselves to the case $\vec z = \vec 0$ for the sake of simplicity.
Assuming that the inequality constraints are non-binding, the solution to a convex optimization problem is found where the gradient of the objective function $\nabla_{\vec f} \OO$ is orthogonal to the linear subspace defined by the equality constraints. This relation is sketched in Fig.~\ref{fig:orthogonality}. In our case, the linear subspace is given by the solutions of the KCL and thus is given by all points 
\begin{align*}
    \vec f = \vec f^{(0)} + \matr C \vec \ell \, 
\end{align*}
where $\matr C \vec \ell $ is an arbitrary cycle flow. \\
We now choose a standard unit basis with basis vectors $\vec u_\alpha, \, \alpha = 1,\ldots,M-N+1$ for the cycle space. The orthogonality condition of the gradient and the linear subspace is then written as
\begin{align*}
    \nabla_{\vec f} \OO(\vec f^\star) 
    \, \cdot \, 
    \matr C \vec u_\alpha = 0,
    \qquad \forall \, \alpha = 1,\ldots,M-N+1.
\end{align*}
Now we can insert the objective function $\OO_{\rm rp}$ or $\OO_{\rm lin}$, compute the gradient and evaluate the condition. The resulting conditions are nothing but the cycle condition
\eqref{eq:cycle-condition2} and the KVL \eqref{eq:KVL2}, respectively. Hence, these graph theoretic conditions have an intuitive geometric interpretation in the context of our optimization problem.

\subsection{Gradient descent}

The geometric interpretation enables another extension of the linear power flow approximation. The linear flows $\vec f^{\rm (lin)}$ provide an easily computable approximation to the real power flows $\vec f^{\rm (rp)}$. 
In Sec.~\ref{sec:improving-dc-newton} we have shown that this approximation can be improved by minimizing $\OO_{\rm rp}$ in the spirit of Newton's method. 
Instead, we may also minimize $\OO_{\rm rp}$ using a gradient descent approach. However, we must keep in mind the geometric aspects of the problem: Any optimization must take place on the affine linear subspace defined by the KCL $\matr E \vec f = \vec p$. Hence, we cannot simply use the  gradient of the objective function $\vec{\nabla} \OO_{\rm rp}$, but rather need its projection onto the linear subspace. Luckily, lemma \ref{lem:helmholtz} shows how to implement this projection. The gradient descent step is thus given by 
\begin{align}
    \vec f' = \vec f - \gamma \, \matr \Pi_{\rm cycle} \, \vec{\nabla}\OO_{\rm rp}(\vec f),
    \label{eq:gradient-step}
\end{align}
where $\gamma$ is the step size. Notably, the gradient is given by
\begin{align}
    \vec{\nabla}\OO_{\rm rp}(\vec f) = \arcsin(\matr K^{-1} \vec f),
\end{align}
where the $\arcsin$ is taken element-wise.

We now provide a numerical example and evaluate the gradient descent step \eqref{eq:gradient-step} for an adapted \textsc{Matpower} 30-bus test case, see Appendix \ref{sec:app:test_case}. The optimal step size $\gamma^\star$ for only one step can be found numerically, cf.~Fig.~\ref{fig:grad_descent}. With one optimal gradient descent step the error between the approximated line flows and the real power flows $\vec f^{\rm(rp)}$ is reduced as expected.  In particular, the error on the heaviest loaded line is reduced. However, the optimal step size depends on the power injections and the topology of the grid, and is not known a priori. 
One may iterate the gradient step, where appropriate values of $\gamma$ may be determined by line search or more advanced methods~\cite{sun2006optimization}.
As an alternative, one may resort to Newton's methods introduced in Sec.~\ref{sec:improving-dc-newton} which provides good results in a single step. 

At this point, we emphasize another application beyond numerical optimization. Given the linear power flow approximation, we can heuristically predict \emph{how} this approximation misses the nonlinear solution. In almost all cases, the first gradient descent step \eqref{eq:gradient-step} will point in the correct direction. That is
$\left[ \matr \Pi_{\rm cycle} \, \vec{\nabla}\OO_{\rm rp}(\vec f^{\rm (lin)}) \right]_e > 0$
typically implies that
$f^{\rm (rp)}_e < f^{\rm (lin)}_e$.
These heuristics may be used to assess the robustness of the linear power flow: Assume that a line $e$ is almost fully loaded in the linear power flow approximation, $f_e^{\rm (lin)} \approx K_e$. The heuristics then show whether the flow $f_e^{\rm (rp)} $ is higher or lower and thus indicates a potential overload.

\subsection{Properties of real and linear power flows: Examples}

\begin{figure}[tb]
\centering
\includegraphics[width=\columnwidth]{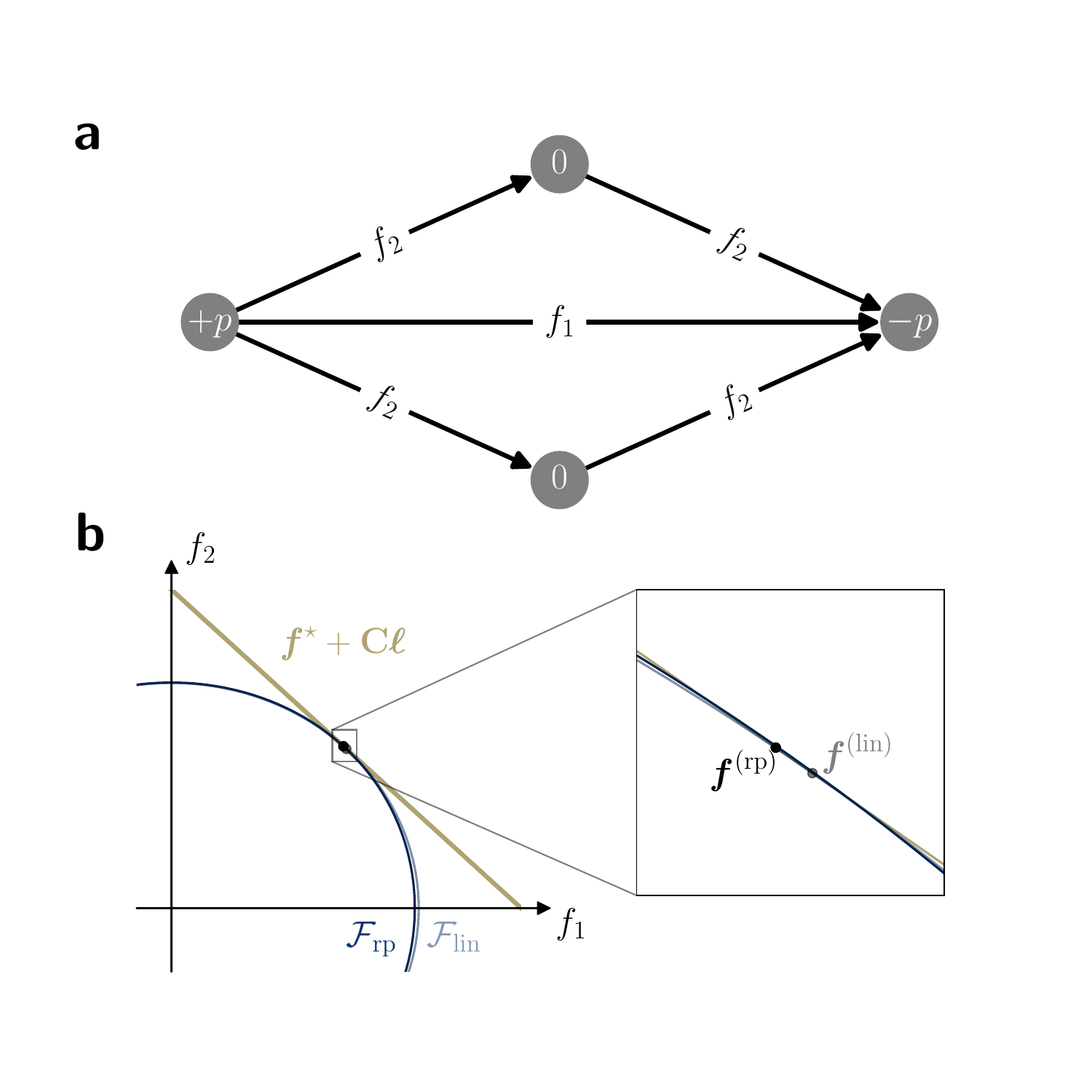}
\caption{
Geometry of the optimization problems associated with the real power flow and the linear power flow approximation. 
(a) We consider an elementary 4-node network with a high degree of symmetry and only two degrees of freedom $f_1$ and $f_2$. 
(b) The contour lines of the objective function $\OO _{lin}(\vec f)$ are ellipses. In comparison, the contour lines of $\OO _{rp}(\vec f)$ are slightly contracted along the axes.
The respective optimizers $\vec f^{(rp)}$ and $\vec f^{(lin)}$ of the constrained optimization problems (dots) are the points where the contour lines are tangent to the affine subspace defined by $\matr E \vec f = \vec p$, cf.~Fig.~\ref{fig:orthogonality}.
Due to the contraction along the axes, the optimizer  $\vec f^{(rp)}$ is further from the coordinate axis. Physically, this corresponds to a more balanced flow and, thus, a smaller maximum line loading.
\label{fig:rp-dc-geometry}
}
\end{figure}

The geometric interpretation of the optimization problem allows further insights into the properties of the power flows, in particular the relation of the real power equations and the linear power flow approximation. We first consider a simple example sketched in Fig.~\ref{fig:rp-dc-geometry}\textbf{a}. For simplicity, we assume that all transmission lines have $K_e=1$. Exploiting the symmetry of the problem, we have only two independent variables $f_1$ and $f_2$ that have to satisfy the constraint
\begin{align}
    f_1 + 2 f_2 = \bar p.
    \label{eq:example-con}
\end{align}
The physical flows are found by minimizing the objective
\begin{align}
\begin{split}
    \label{eq:example-con_rp}
    \OO_{\rm rp}(f_1,f_2)  = & f_1 \arcsin(f_1) + \sqrt{1-f_1^2} \\
    & + 4  f_2 \arcsin(f_2)  
    + 4 \sqrt{1-f_2^2} 
    - 5 .
\end{split}
\end{align}
If we would instead invoke the linear power flow equation, we would have
\begin{align}
    \label{eq:example-con_lin}
    \OO_{\rm lin}(f_1,f_2) = f_1^2 + 4 f_2^2 .
\end{align}
The two optimization problems are illustrated in Fig.~\ref{fig:rp-dc-geometry}\textbf{b}. Notably, the constraint \eqref{eq:example-con} defines an affine linear subspace which is shown as a straight line in the figure.
The minimizer is found where this affine subspace is \emph{tangent} to the surface of constant $\OO(f_1,f_2)$. Comparing real and linear power, we see that the surface of constant $\OO_{\rm rp}$ is more `angled' than the surface of constant $\OO_{\rm lin}$. Hence, the minimizer $\vec f_{\rm lin}$ is closer to one of the axes than the minimizer $\vec f_{\rm rp}$. That is, the real power flows $f_1$ and $f_2$ are more balanced than predicted by the linear power flow approximation.
Similar geometric arguments apply to many real or linear power flow problems. We conclude by comparing the real power flow to the linear power flow approximation that flows are \emph{typically} more evenly distributed and the maximum loading is lower. We note that this topic has been previously addressed in Ref.~\cite{dorfler2013synchronization}, leading to an efficient synchronization condition.

\begin{figure}[tb]
\includegraphics[width=\columnwidth]{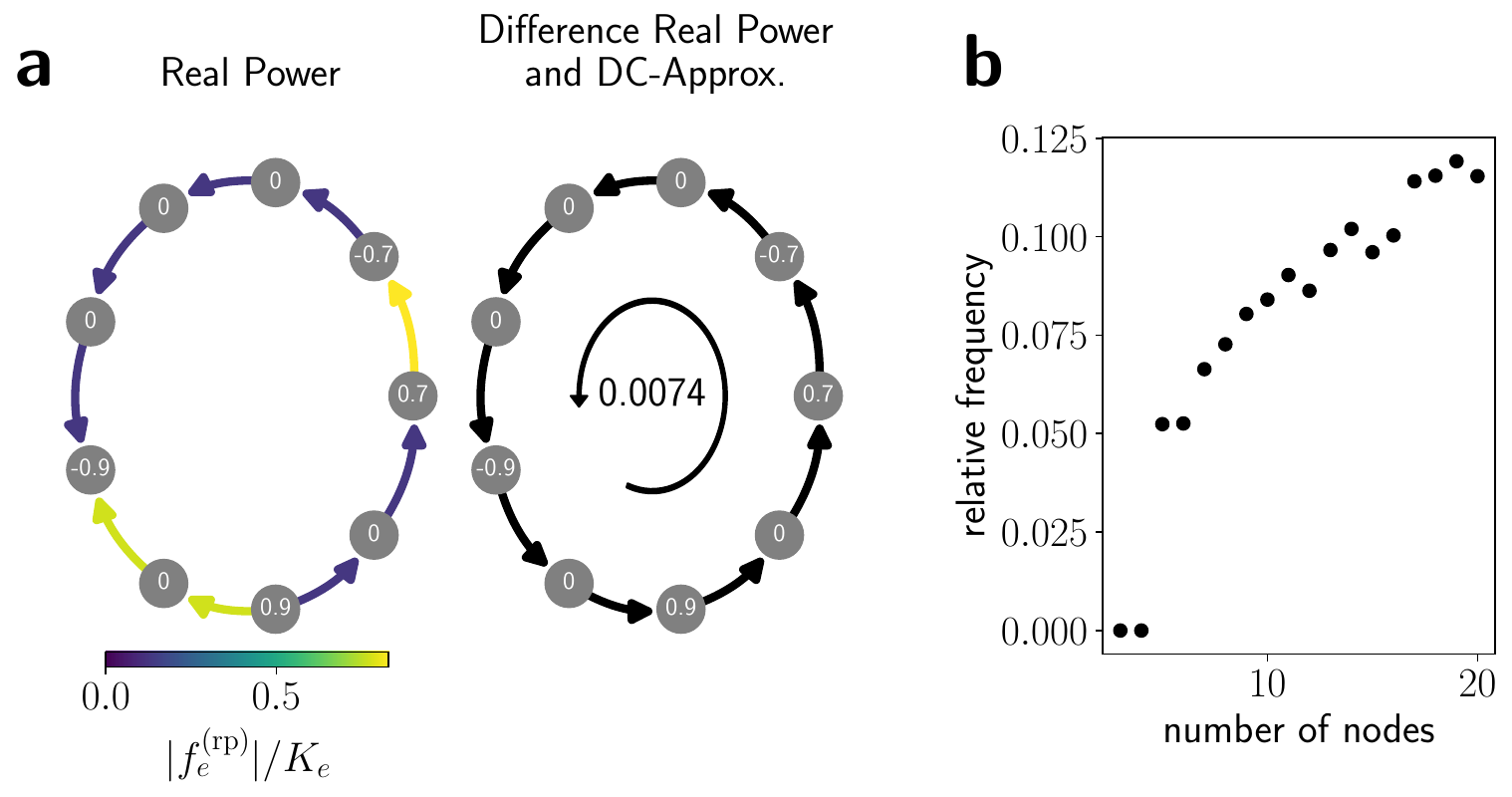}
\caption{
The maximum line loading of the real power flow and the linear power flow approximation in a cyclic network.
(a) An elementary example with $N=9$ nodes and homogeneous coupling strengths $K_e \equiv 1$ and power injections $p_n$ given in cycles.
The left map shows the loading $|f_e^{\rm (rp)}|/K_e$ and the direction of the flow for each edge $e$. The difference $f_e^{(\rm rp)} - f_e^{(\rm lin)}$ can be expressed as a cycle flow $\ell$ shown in the right map.
Contrary to the geometric intuition developed in Fig.~\ref{fig:rp-dc-geometry}, the linear approximation underestimates the maximum line loading $\max_e |f_e|/K_e$.
That is, the cycle flow $\ell$ is parallel to the flow on the most heavily loaded line.
(b) For a systematic comparison, we consider networks of different sizes $N$ and $10^6$ randomly sampled and valid power injection vectors $\vec p$. For each network we compute the real power flow $\vec f^{\rm (rp)}$
and the linear approximation $\vec f^{\rm (lin)}$.
For each $N$ we count the number of networks for which 
$\max_e |f_e^{(\rm rp)}|$ exceeds $\max_e |f_e^{(\rm lin)}|$ by at least $10^{-4} \%$.
The relative frequency of such networks increases with the size $N$, reaching approx. $12\%$ for $N = 20$. No such networks are observed for $N\le 4$. 
\label{fig:maxload-cycles}
}
\end{figure} 

However, one must  be careful in formalizing and generalizing these statements. While they are often true, they do not hold in every case. To gain further insights, we will conduct a numerical experiment comparing the maximum load $\max_e |f_e|/k_e$  in the real power flow and in the linear approximation. For this experiment, we investigate cyclic networks with $N$ nodes and all line parameters set to unity, $K_e \equiv 1 \forall e \in \EE$. A first example is shown in Fig.~\ref{fig:maxload-cycles}a. In this case, we find that the maximum load $\max_e |f_e|/k_e$ is \emph{underestimated} in the linear power flow approximation contrary. This appears surprising but does not contradict our general reasoning. 
The network includes three very highly loaded edges. The linear power flow overestimates the flow for two of them and underestimates it for only one. Hence, we find that the linear power flow overestimates high loadings \emph{on average}, while we cannot make a general statement for individual lines. We will provide a rigorous statement in lemma \ref{lew:heavyloadedlines}. \\
We continue our numerical experiment, considering different network sizes $N$ and use cases. We uniformly sample power injections $\vec p$ such that the grid is balanced $\sum_n p_n = 0$ and compute the real power flow $\vec f^{\rm (rp)}$ as well as the linear approximation $\vec f^{\rm (lin)}$. For each value of $N$ we count the the number of samples where the linear approximation underestimates the maximum loading, that is $\max_e |f_e^{\rm (rp)}|  > \max_e |f_e^{\rm (lin)}| $. Results are shown in Fig.~\ref{lew:heavyloadedlines}b.
For $N \leq 4$ we find no cases where the linear approximation underestimates the maximum load. The number of cases increases with the cycle length $N$ until it reaches about $12\%$ for a ring with $N=20$ nodes. Notably, cycles with $N \leq 4$ are special as they do not support multistability~\cite{manik2017cycle} and admit explicit stability conditions~\cite{dorfler2013synchronization}.

\subsection{Properties of real and linear power flows: Rigorous results}

We now give two rigorous result on the distribution of line loading and the maximum line loading comparing real and linear power flow. We first show that there are more heavily loaded lines in linear power flow approximation than in the nonlinear real power flow. To make this statement rigorous, we introduce a function that indicates a heavy loading. Extending theorem
\ref{thm:david1} we define the function
\begin{align}
    \label{eq:indcator_line_loading}
    \hat \GC_e(f_e) = \GC_e(f_e) / \GC_e(K_e),
\end{align}
which is plotted in Fig.~\ref{fig:fig:indicator_line_loadings}. This function increases monotonously and nonlinearly with the line load $|f_e| / K_e$. If a line is weakly loaded, $|f_e| < K_e/2$, then the function is close to zero $\hat \GC_e(f_e) < 0.04$. If the function is heavily loaded $|f_e| \approx K_e$, the function approaches unity. We can thus interpret the function $ \hat \GC_e(f_e)$ as the desired indicator for heavy loading. We can then find the following statement showing that the number of heavily loaded lines is smaller for $\vec f^{\rm (rp)}$ than for $\vec f^{\rm (lin)}$
\begin{lem}
\label{lew:heavyloadedlines}
The weighted sum of heavily loaded lines as measured by the indicator function $\hat \GC_e(f_e)$ satisfies
\begin{align}   
   \sum_e K_e \, \hat \GC_e(f^{\rm (rp)}_e)  
   \le 
   \sum_e K_e \, \hat \GC_e(f^{\rm (lin)}_e) 
   .
\end{align}
\end{lem}
\begin{proof}
We use the notation of theorem \ref{thm:david1}.
The vector $\vec f^{\rm (rp)}$ minimizes the objective function $\OO_{\rm rp}$. Hence we obtain
\begin{align*}
    \frac{1}{2} \|  \vec f^{\rm (rp)}  \|_K^2 + \GC(\vec f^{\rm (rp)}) 
    \le
     \frac{1}{2} \|  \vec f^{\rm (lin)}  \|_K^2 + \GC(\vec f^{\rm (lin)}) 
\end{align*}
Using $\|  \vec f^{\rm (rp)}  \|_K^2 = \|  \vec f^{\rm (lin)}  \|_K^2 + \|  \vec \xi  \|_K^2$, this is rewritten as
\begin{align*}
    \GC(\vec f^{\rm (rp)}) 
    \le
     \GC(\vec f^{\rm (lin)}) -  \frac{1}{2} \|  \vec \xi  \|_K^2 \leq  \GC(\vec f^{\rm (lin)})
\end{align*}
This inequality can be rewritten in components as
\begin{align*}
     \sum_e  \GC(\vec f_e^{\rm (rp)})
     &\le 
     \sum_e \GC(\vec f_e^{\rm (lin)}) \\
     \Leftrightarrow \quad \sum_e K_e \, \frac{\GC(\vec f_e^{\rm (rp)})}{( \pi-1) K_e/2}
     &\le 
     \sum_e K_e \, \frac{\GC(\vec f_e^{\rm (lin)})}{( \pi-1) K_e/2}
\end{align*}
Using $\GC(K_e) = (\pi-1) K_e/2$ this yields the desired result.
\end{proof}

\begin{figure}[tb]
\includegraphics[width=\columnwidth]{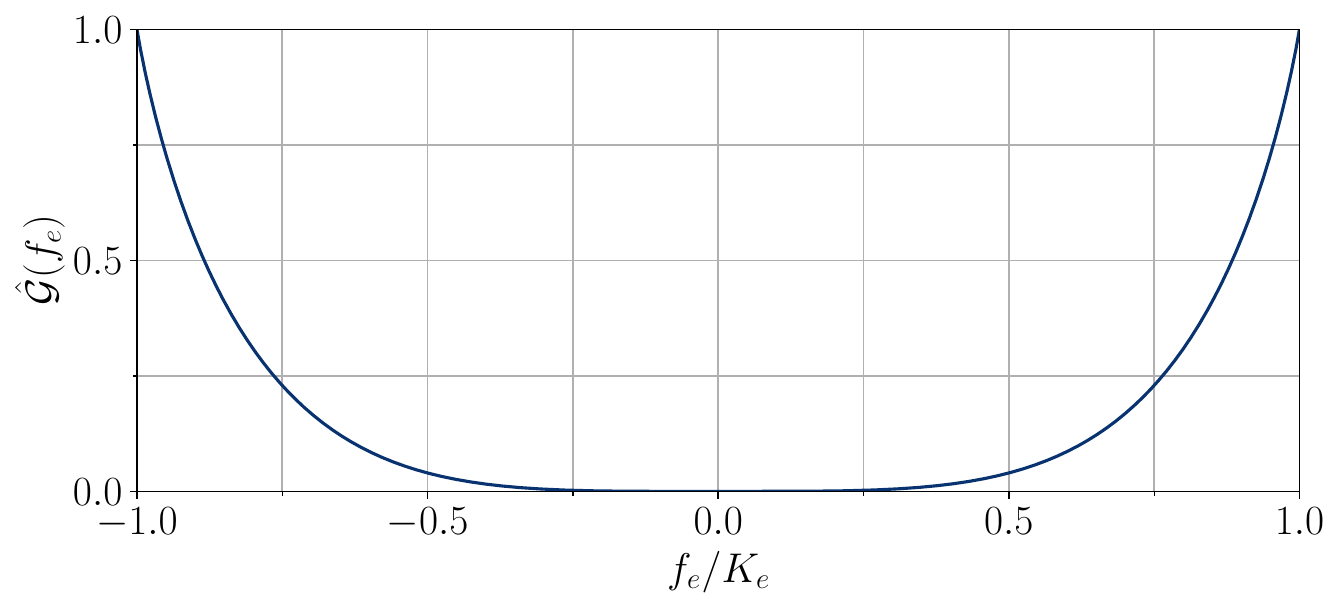}
\caption{
The indicator function $\hat \GC_e(f_e)$ in dependence of the line loading $f_e/K_e$, see Eq.~\eqref{eq:indcator_line_loading}.
\label{fig:fig:indicator_line_loadings}
}
\end{figure}

As a second step, we consider the two most heavily loaded lines within any cycle of the network. We prove that the ratio of the line loadings $|f_e/K_e|$ is bounded, where the bound is tighter for the real power flow than for the linear power flow. That is, the line loading is more homogeneous. 

\begin{lem}
\label{lem:homogeneity1}
Let $\vec f^{\rm (rp)}$ be the solution of the real power flow equations with $\vec z=\vec 0$ and $\vec f^{\rm (lin)}$ the solution of the linear power flow equations.
For every cycle $\mathcal{C}$ in the network we consider the two most heavily loaded lines $a$ and $b$, i.e.
\begin{align}
     \frac{|f_a|}{K_a} \ge 
     \frac{|f_b|}{K_b} \ge
     \frac{|f_e|}{K_e}
     \qquad 
     \forall \, e \in \mathcal{C} \backslash \{a,b\} \, .
\end{align}
Then 
\begin{align*}
    \frac{ \arcsin |f^{\rm (rp)}_a/K_a| 
        }{ \arcsin |f^{\rm (rp)}_b/K_b| }
      &\le
      |\mathcal{C}| -1, \\
    \frac{ |f^{\rm (lin)}_a/K_a| 
        }{ |f^{\rm (lin)}_b/K_b| }
      &\le
      |\mathcal{C}| -1.
\end{align*}
\end{lem}

\begin{proof}
We choose a cycle basis such that the cycle $\mathcal{C}$ corresponds to the first fundamental cycle, i.e.,~the first row of the edge cycle incidence matrix. For the minimizer $\vec f^{\rm (rp)}$ we have
 \begin{align*}
     0 = \frac{\partial \OO_{\vec z = \vec 0}}{\partial \ell_1} = \sum_{e \in \mathcal{C}} 
     C_{1,e}
     \arcsin \left( \frac{f_e^{\rm (rp)}}{K_e}   \right).
\end{align*}
Now we can bound the different terms in the sum and obtain
\begin{align*}
     \arcsin \left( \frac{|f_a^{\rm (rp)}|}{K_a}   \right) 
     &= \left| \sum_{e \in \mathcal{C} \backslash \{a\} } 
     C_{1,e}
     \arcsin \left( \frac{f_e^{\rm (rp)}}{K_e}   \right)  \right|   \\
     &\le 
     \left( |\mathcal{C}|-1 \right) \times
     \arcsin \left( \frac{|f_b^{\rm (rp)}|}{K_b}   \right) .
\end{align*}
Replacing $\OO_{\vec z = \vec 0}$ by $\OO_{\rm lin}$ yields the corresponding inequality for the linear power flow solution $\vec f^{\rm (lin)}$.
\end{proof}

This result is particularly useful if we deal with the stationary states of a grid after applying a Kron reduction (cf.~Sec.~\ref{sec:background}). Then the network is fully connected, and we can choose the fundamental cycles to be triangles such that $|\mathcal{C}|-1 = 2$. We then obtain the following example from lemma \ref{lem:homogeneity1}. Assume that a line $(n,m)$ is loaded with $f^{\rm (rp)}_{nm}/K_{nm} \ge 99 \%$. Then for every vertex $l \neq n,m$ one of the lines $(l,m)$ or $(l,n)$ must be loaded with $f^{\rm (rp)}/K \ge 66 \%$. That is, it is impossible for one line to get very heavily loaded in isolation -- some lines in the vicinity must be heavily loaded, too.

\section{Conclusion and Outlook}
\label{sec:conclusion}

The reliable supply of electric power is material for our society. Therefore, it is of central importance to understand which factors determine the operation and stability of the electric power system. In this article, we have analyzed the stationary states of lossless power grids described by the equations
\begin{align}
    p_n = \sum_{m} K_{n,m} \sin(\theta_n - \theta_m).
\end{align}
Being nonlinear, there is no general theory of solvability of this set of equations. Furthermore, it has been established that the equations can be multistable depending on the structure of the grid~\cite{manik2017cycle,delabays2016multistability}.

In this article, we have introduced a novel approach to the stationary states of lossless power grids. The main idea is to shift the attention from the nodes of the grid to the edges and cycles and to reformulate the equations as a convex optimization problem. This formulation provides new insights into the structure of the problem and allows to derive a series of rigorous results. 

The most important results are as follows.
\begin{enumerate}
\item The optimization approach provides a systematic algorithm to systematically compute \emph{all} stationary states that satisfy $|\theta_n - \theta_m| < \pi/2$ for all edges $(n,m) \in \EE$.
\item The linear power flow or DC approximation is recovered as the quadratic approximation to the optimization problem. This insight allows to systematically bound the error induced by this approximation.
\item We have introduced two explicit formulae that provide improved approximations for the real power flows.
\item The optimization approach provides a geometric interpretation of the network equations. This interpretation provides rigorous results on the existence of solutions.   
\end{enumerate}

The basic idea of this article can be generalized to related network flow problems. 
First, a generalization beyond the sinusoidal coupling $\sin(\theta_n - \theta_m)$ is straightforward as long as the coupling is anti-symmetric and invertible on a certain interval. 
That is, the set of equations
\begin{align}
    p_n = \sum_{m} K_{n,m} h(\theta_n - \theta_m)
\end{align}
with an anti-symmetric function $h$ can be reformulated as an optimization problem with the objective function 
\begin{align}
    \OO(\vec f) = \sum_{e \in \EE} H(f_e/K_e),
\end{align}
where $H$ is the primitive of the inverse $h^{-1}$. Unfortunately, there is no obvious way how Ohmic losses or not anti-symmetric coupling functions may be included in the formalism.

Second, the basic idea may be generalized to the lossless load-flow equations. This set of equations includes the voltage magnitudes and reactive power flows in addition to the real power. Rigorous results on the solvability of the load-flow equations are notoriously difficult to obtain but of outstanding practical importance.

\acknowledgements

We thank Andrea Benigni, Thiemo Pesch, Manuel Dahmen and Lukas Kinzkofer for stimulating discussions.
The authors gratefully acknowledge support from the Deutsche Forschungsgemeinschaft (DFG, German Research Foundation) via Grant No.~491111487.

\appendix

\section{Solutions of the KCL}
\label{app:kcl}

In this appendix we review results on the existence of solutions of the KCL with with flow constraints,
\begin{align}
    \matr E \vec f = \vec p,
    \qquad
    |f_e| \le K_e 
    \quad
    \forall e \in \EE.
    \label{eq:kcl-con-app}
\end{align}
In the context of the optimization problems \eqref{opt:realpower1} and \eqref{opt:realpower2-cycle}, this is equivalent to the question whether the feasible set is non-empty. We provide two lemmas to systematically answer this question following Ref.~\cite{manik2017cycle}.

First, one can map the given problem with multiple sources and sinks to a  single-source single-target flow problem which is commonly studied in graph theory~\cite{ford2015flows,nussbaum2010multiple}. 
Given a network with vertex set $\VV$, edge set $\EE$, and edge capacities $K_e$, we define an extended graph $\GG' = (\VV',\EE')$ by adding two vertices $s$ and $t$, that are connected to the sources and sinks, respectively. That is,
\begin{align*}
    \VV' &= \VV \cup \{s,t\}, \\
    \EE' &= \EE \; \cup \; \{ (s,n) | p_n > 0 \}
               \; \cup \; \{ (n,t) | p_n < 0 \}. 
\end{align*}
The coupling strengths of the new edges $K_{s,n}$ and $K_{t,n}$ are infinite. Furthermore, we define the cumulative input and output power
\begin{align*}
    \hat p_s = \sum_{n \in \VV, p_n>0} p_n \\
    \hat p_t = \sum_{n \in \VV, p_n<0} p_n .
\end{align*}
As usual, we assume that the network is balanced such that $\hat p_s = - \hat p_t$.

\begin{lem}
A solution of the  KCL with line limits \eqref{eq:kcl-con-app} exists if and only if the maximum $s-t-$flow in the extended network $\GG'$ is larger or equal to $\hat p_s$.
\end{lem}

Second, we give a criterion in terms of partitions of the network. Let $(\VV_1,\VV_2)$ be an arbitrary partition of the vertex set $\VV$ such that
\begin{align*}
    \VV_1 \cup \VV_2 = \VV 
    \qquad \mbox{and} \qquad
    \VV_1 \cap \VV_2 = \emptyset 
\end{align*}
and $\EE(\VV_1,\VV_2) = \{ (n,m) \in \EE | n \in \VV_1, m \in  \VV_2  \}$ be the cut-set induced by this partition. Define 
\begin{align*}
    \bar p_1 = \sum_{n \in \VV_1} p_n, \quad
    \bar p_2 = \sum_{n \in \VV_2} p_n, \quad 
    \bar K_{12} = \sum_{e \in \EE(\VV_1,\VV_2)} K_{e}.
\end{align*}
Then we have the following lemma (see Ref.~\cite{manik2017cycle} for a proof).

\begin{lem}
If for \emph{all} partitions $(\VV_1,\VV_2)$ we have
\begin{align}
    |\bar p_1 | = |\bar p_2 | \le \bar K_{12},
    \label{eq:kcl-lemma2-condition}
\end{align}
then there exists a solution of the KCL with line constraints \eqref{eq:kcl-con-app}.
\end{lem}

\section{Proof of Lemma \ref{lem:cycle_flows_bound_k_norm}}
\label{app:proof-norm-max-cycle}

To prove the inequality 
\begin{equation*}
\| \vec f^{(c)} \|_K^2 \geq \big( f^{(c)}_a \big)^2 \frac{1}{K_a (1 - K_a \Omega_a)}
\end{equation*}
for each cycle flow $\vec f^{(c)}$ and for each edge $a = (n,m) \in \EE$, we proceed in reverse order. That is, we fix the flow on one edge $a$ and solve the optimization problem
\begin{align*}
    &\min_{\vec g} \|\vec g\|_K^2 \\
    \mbox{s.t. } & \matr E \vec g = 0, \quad g_a = f^{(c)}_a
\end{align*}
by the methods of Lagrangian multipliers. We define the Lagrangian
\begin{equation*}
    \mathcal{L}( \vec g) = \|\vec g\|_K^2 - \sum_n \bigg[ \lambda_n \sum_b E_{n,b} g_b \bigg] - \mu (g_a - f_a^{(c)})
\end{equation*}
and find the stationary points 
\begin{align*}
    g_{(n,m)} = \begin{cases}
        K_{nm} (\lambda_n - \lambda_m), \quad & \text{if } (n,m) \neq a\\
        K_{nm} (\lambda_n - \lambda_m + \mu) \quad & \text{if } (n,m) = a
\end{cases}
\end{align*}
and note that $f^{(c)}_{a = (n,m)} = K_{nm} (\lambda_n - \lambda_m + \mu)$. 
Hence, the optimizer $\vec g^{\star}$ generally assumes the form
of a potential flow except for the edge $a$. Hence we can write it as
\begin{equation*}
    \vec g^{\star} = \matr K \matr E^\top \vec \lambda + c \vec w_a, 
\end{equation*}
where again $\vec w_a$ is the $a$-th standard basis vector. We now have to determine $\vec \lambda$ and $c$ such that the two constraints are satisfied.
\begin{enumerate}
    \item[1.] Evaluating the first constraint $\matr E \vec g^{\star} = 0$ yields
    \begin{align*}
        \matr E \matr K \matr E^\top \vec \lambda + c \matr E \vec w_a = 0. \\
    \end{align*}  
    Using that $\matr E \matr K \matr E^\top =\matr L$ this set of linear equations is solved by
    \begin{align*}
        \vec \lambda = - c \matr L^+ \matr E \vec w_a \, .
    \end{align*}
    \item[2.] The second constraint $g_a = f_a^{(c)}$ can alternatively be written as $f_a^{(c)} = \vec w_a^\top \vec g^{\star}$ and we thus have
    \begin{align*}
        f_a^{(c)} & = \vec w_a^\top \left( c \vec w_a +  \matr K \matr E^\top  (- c \matr L^+ \matr E \vec w_a) \right) \\
        \Leftrightarrow	 \quad   c &= f_a^{(c)} \left( 1 - \vec w_a^\top \matr K \matr E^\top \matr L^+ \matr E \vec w_a \right)^{-1}.
    \end{align*}
\end{enumerate}
We can now compute
\begin{align*}
    \lVert \vec g^{\star} \rVert_K ^2 &= \vec g^{\star \top} \matr K^{-1} \vec g^{\star} \\
    &= \left(c \vec w_a^\top + \vec \lambda^\top \matr E \matr K  \right) \matr K^{-1} \left( \matr K \matr E^\top \vec \lambda + c \vec w_a \right) \\
    &= c^2 K_a^{-1} + \vec \lambda^\top \matr L \vec \lambda + 2c \vec w_a^\top \matr E^\top  \vec \lambda.
\end{align*}
Inserting the expressions for $\vec \lambda$ yields
\begin{align*}
    \lVert \vec g^{\star} \rVert_K ^2 &= c^2 K_a^{-1} 
    + c^2 \vec w_a^\top \matr E^\top \underbrace{\matr L^{+} \matr L  \matr L^{+}}_{=\matr L^{+}} \matr E \vec w_a \\
    & \quad \quad - 2c^2 \vec w_a^\top \matr E^\top \matr L^+ \matr E \vec w_a \\
    &= c^2 \left( K_a^{-1} - \vec w_a^\top \matr E^\top \matr L^+ \matr E \vec w_a \right).
\end{align*}
Thus, inserting the expression for $c$ and using that $\vec w_a^\top \matr K = K_a \vec w_a^\top$ we get
\begin{align*}
    \lVert \vec g^{\star} \rVert_K ^2 &= (f_a^{(c)})^2 \left( 1 - \vec w_a^\top \matr K \matr E^\top \matr L^+ \matr E \vec w_a \right)^{-2} \\
    & \quad \quad \cdot \left( K_a^{-1} - \vec w_a^\top \matr E^\top \matr L^+ \matr E \vec w_a \right) \\
    &= (f_a^{(c)})^2 K_a^{-1} \left( 1 - K_a \underbrace{\vec w_a^\top  \matr E^\top \matr L^+ \matr E \vec w_a}_{=\Omega_a} \right)^{-1},
\end{align*}
which concludes the proof by noting that $\lVert \vec g^{\star} \rVert_K ^2 \leq \| \vec f^{(c)} \|_K^2$. 
\qed 

\section{Additional material for the proof of theorem \ref{thm:david1}}
\label{app:taylor}

In this appendix, we provide a technical result used in the proof of theorem \ref{thm:david1}.
We use the formulation \eqref{eq:objective-z} of the objective function with $\vec z = \vec 0$, $\vec f^{(0)} = \vec f^{\rm (rp)}$ and write $\vec \xi = - \matr C \vec \ell$. Then
\begin{align*}
      \OO_{\rm rp}(\vec f^{\rm (lin)})
    &= \OO_{\vec z = \vec 0}(-\vec \ell), \\ 
    \OO_{\rm rp}(\vec f^{\rm (rp)})
    &= \OO_{\vec z = \vec 0}(\vec 0).
\end{align*}
Now we apply the multivariate Taylor theorem with the Lagrange form of the remainder which yields
\begin{align*}
    \OO_{\vec z = \vec 0}(-\vec \ell) = & 
    \OO_{\vec z = \vec 0}(\vec 0)
    - \nabla \OO_{\vec z = \vec 0}(\vec 0) \cdot \vec \ell \\
    & + \frac{1}{2}
    \vec \ell^\top \nabla^2 \OO_{\vec z = \vec 0}(\vec \lambda) \vec \ell,
\end{align*}
where $\vec \lambda \in \mathbb{R}^{\nl-\nn+1}$ is a vector with entries $\lambda_\alpha$ between $0$ and $\ell_\alpha$. The linear term of the Taylor expansion vanishes because we expand around the minimizer. The quadratic term can be bounded from below. Using the Hessian \eqref{eqn:Hesse-ell} we obtain
\begin{align*}
    & \vec \ell^\top \nabla^2 \OO_{\vec z = \vec 0}(\vec \lambda) \vec \ell \\
    &= \vec \ell^\top \matr{C}^\top
    \mbox{diag} \left( \frac{1}{\sqrt{K_e^2 - (f^{(0)}_{e} + \sum_\varphi  C_{e \varphi} \lambda_\varphi )^2}}
    \right) \matr{C} \vec \ell \\
    &\ge \vec \xi^\top
    \mbox{diag} \left( \frac{1}{K_e}
    \right)
    \vec \xi  = \| \vec \xi \|_K^2 \, .
\end{align*}
Hence we obtain
\begin{align*}
     \OO_{\rm rp}(\vec f^{\rm (lin)})
     \ge \OO_{\rm rp}(\vec f^{\rm (rp)}) + \frac{1}{2}  \| \vec \xi \|_K^2.
\end{align*}

\section{\textsc{Matpower} 30-bus test case}
\label{sec:app:test_case}

The original \textsc{Matpower} 30-bus test case \cite{matpowercase30} data is intended to be used to study the full AC power flow. Hence, the power injections are only balanced up to losses. In this work, we neglect losses as well as the reactive power flows, and thus slightly adapt the \textsc{Matpower} 30-bus test case to our needs. We keep the topology and line admittances of the grid, but re-balance the real power injections. That is, the real power imbalance $\sum_n p_n$ of the loads and generators is added to the power injection of the first generator to arrive at a balanced grid. 
To study different grid loads, we introduce a power factor $p_f$ as a multiplicative scalar that rescales the re-balanced power injections, effectively changing the line loading.


%

\end{document}